%% file: MT_sampler-manuscript.tex
\documentclass[10pt,journal,letterpaper]{IEEEtran}
\usepackage[T1]{fontenc}
\usepackage{amsmath,amssymb,amsfonts,amsthm}
\usepackage{graphicx}
\usepackage{textcomp,nicefrac}
\usepackage{bm}
\usepackage{microtype}
\usepackage[hyphens]{url}
\usepackage{needspace}

\DeclareFontShape{T1}{ptm}{m}{scit}{<->ssub * ptm/m/it}{}

\newtheorem{definition}{Definition}

\newtheorem{theorem}{Theorem}
\newtheorem{corollary}{Corollary}

\theoremstyle{remark}

\theoremstyle{plain}

\newcommand{\R}{\mathbb{R}}

\newcommand{\Prob}{\mathbb{P}}
\newcommand{\Exp}{\mathbb{E}}
\newcommand{\ind}{\mathbf{1}}
\newcommand{\B}{\mathcal{B}}

\begin{document}
\setlength{\abovedisplayskip}{5.25pt}
\setlength{\belowdisplayskip}{\abovedisplayskip}
\setlength{\abovedisplayshortskip}{\abovedisplayskip}
\setlength{\belowdisplayshortskip}{\abovedisplayskip}
\raggedbottom
\bstctlcite{IEEE:BSTcontrol}
\clubpenalty=10000
\widowpenalty=10000
\displaywidowpenalty=10000
\interdisplaylinepenalty=10000
\AddToHook{env/definition/before}{\Needspace{4\baselineskip}}
\AddToHook{env/proof/before}{\Needspace{3\baselineskip}}
\setlength{\emergencystretch}{5em}

\title{Multi-Threshold Sampling: Signal Space, Sampling Operators, and First Recorded Crossing Time}

\author{Ao Qiu
  and Qingguo Xie
  \thanks{This work was supported in part by the National Natural Science Foundation of China under Grant 625B2079, Grant 61927801, Grant 62250002, and Grant 62050288. \textit{(Corresponding author: Qingguo Xie.)}}%
  \thanks{Ao Qiu is with the Department of Biomedical Engineering, Huazhong University of Science and Technology, Wuhan 430074 China (e-mail: aqiu@hust.edu.cn).}%
  \thanks{Qingguo Xie is with the Department of Biomedical Engineering, Huazhong University of Science and Technology, Wuhan 430074 China; with the Wuhan National Laboratory for Optoelectronics, Wuhan 430074 China; and with the Department of Electronic Engineering and Information Science, University of Science and Technology of China, Hefei 230026 China (e-mail: qgxie@hust.edu.cn).}}

\maketitle

\begin{abstract}
Multi-threshold sampling encodes signals through crossing times at prescribed thresholds, offering an alternative to uniform time-domain sampling. Its potential to reduce hardware requirements and power consumption motivates its use in high-speed pulse acquisition, including scintillation detector readout. We model the relation between signal statistics and recorded-crossing statistics by composing measurable stochastic crossing, timing, and selection operators. The threshold noise processes are independent across thresholds. The sampler applies independent timing jitter at each analog crossing before quantization and enforces dead time separately at each threshold in the recording window. With timing jitter shared within each threshold, path infima and suprema give the joint cumulative distribution function of first recorded crossing times if, for each selected threshold and its specified cutoff, the signal almost surely neither coincides with the threshold on an interval of positive length nor changes sides at either endpoint of the corresponding shifted interval. Specializing this formula to a single threshold gives the marginal distribution, from which we obtain the conditional standard deviation. We use these formulas to approximate the corresponding exact quantities under independent timing jitter at each crossing and bound each approximation error. In a scintillation example, the shared-jitter formulas and independent-jitter waveform simulations agreed within Monte Carlo sampling uncertainty for the joint and marginal distributions and for the conditional standard deviation. Engineers can use these results to choose thresholds and assess how sampler changes affect timing precision.
\end{abstract}


\section{Introduction}
\label{sec:introduction}

Multi-threshold (MT) sampling compares a signal with multiple prescribed amplitude thresholds and records the times at which each threshold is crossed~\cite{xi2013fpga,xie2009potentials,kim2009multi,qiu2026prior}. MT sampling therefore fixes nominal amplitude levels and measures crossing times~\cite{sayiner1996level,rzepka2017bandwidth}, whereas uniform time-domain sampling fixes regularly spaced sampling times and measures signal amplitudes~\cite{nyquist1928certain,shannon1949communication,jerri1977shannon,unser2000sampling}. For fast pulses whose relevant information can be recovered from a few threshold crossings, MT sampling can reduce hardware requirements and power consumption by replacing high-rate amplitude digitization with threshold comparisons and time measurements~\cite{palka2017multichannel}. These advantages make MT sampling attractive for multichannel detector readout, with applications in positron emission tomography~\cite{xie2005new,zhang2025performance}, oil well logging~\cite{zhao2023fpga}, and photon-counting X-ray imaging~\cite{zhang2024feasibility,chen2023multi}.

These applications create a practical need to predict MT sampling performance and optimize sampler parameters. Existing studies address performance assessment through numerical parameter scans~\cite{derenzo2014fundamental} and model-specific approximations~\cite{seifert2012comprehensive}; a unified probabilistic model linking signal statistics to recorded-crossing statistics is still needed. The challenge is to model signals as stochastic processes and account for the physical sequence of comparator decisions, time-to-digital conversion, and selective output.

We model MT sampling by composing measurable stochastic crossing, timing, and selection operators. For a random signal sampled as in Definition~\ref{def:shared_noise_mt_sampler} with timing jitter shared within each threshold, Theorem~\ref{thm:joint_recorded_cdf}(i) gives the joint cumulative distribution function (CDF) of first recorded crossing times from path infima and suprema. This formula holds if, for each selected threshold and its specified cutoff, the signal almost surely neither coincides with the threshold on an interval of positive length nor changes sides at either endpoint of the corresponding shifted interval. The marginal CDF determines the standard deviation (SD) conditional on recording. We bound the errors in both CDFs and the conditional SD when jitter is independent across analog crossings. In a scintillation example, both CDFs and the conditional SD agree with independent-jitter waveform simulations within Monte Carlo sampling uncertainty. Across the tested thresholds, we quantify changes in timing precision and identify the smallest estimated conditional SD.

\section{Signal Space}
\label{sec:signal_space}

We represent signal realizations by paths on a finite observation window, allowing continuous evolution and jumps.

\begin{definition}[Signal path and random-signal spaces]
\label{def:signal_path_space}
Let \(\mathcal{T}\triangleq[T_-,T_+]\subset\R\), with \(T_-<T_+\), be a fixed compact observation window for signal and threshold paths. Let \(\mathsf X\triangleq D(\mathcal T;\R)\) be the signal path space, consisting of real-valued paths that are right-continuous on \([T_-,T_+)\) and have finite left limits on \((T_-,T_+]\). Each such path is bounded on \(\mathcal T\). Write \(\lVert x\rVert_\infty\triangleq\sup_{t\in\mathcal T}|x(t)|\) for the uniform norm. Equip \(\mathsf X\) with the coordinate \(\sigma\)-algebra \(\B(\mathsf X)\triangleq\sigma\{x\mapsto x(t):t\in\mathcal T\}\).

On a complete probability space \((\Omega,\mathcal F,\Prob)\), define the random-signal space
\begin{equation*}
  \mathfrak Y \triangleq \bigl\{[Y]: Y:\Omega\to\mathsf X\text{ is measurable},\Exp\lVert Y\rVert_\infty<\infty \bigr\},
\end{equation*}
where \([Y]\) is the class of processes \(Z\) satisfying \(\Prob\{\lVert Y-Z\rVert_\infty=0\}=1\).
\end{definition}

We write \(Y\) for \([Y]\) and use a measurable representative. For each \(t\in\mathcal T\), \(Y(t)\) denotes the random variable \(\omega\mapsto Y(\omega)(t)\).

The Skorokhod \(J_1\) topology makes \(\mathsf X\) Polish, with Borel \(\sigma\)-algebra \(\B(\mathsf X)\). Thus \((\mathsf X,\B(\mathsf X))\) is standard Borel~\cite[Sec.~12]{billingsley1999convergence}.

\section{Sampling Operators}
\label{sec:sampling_operators}

\subsection{Representation of Crossings}
\label{subsec:crossing_representation}

An \emph{analog}, \emph{timed}, or \emph{recorded} crossing is the output of the crossing, timing, or selection stage, respectively. At each stage a finite point measure records the time, threshold index, and direction of every crossing. The mass at a point gives its multiplicity, and the total mass counts all crossings~\cite{last2018lectures}.

For an integer \(N\geq1\), let \([N]\triangleq\{1,\ldots,N\}\) be the set of threshold indices and \(\Sigma\triangleq\{-1,-\tfrac12,0,\tfrac12,1\}\) the set of direction marks. Equip the space of triples \((t,n,\sigma)\), \(\mathsf E\triangleq\R\times[N]\times\Sigma\), with its product Borel \(\sigma\)-algebra. For \(z\in\mathsf E\), let \(\delta_z\) denote the Dirac measure at \(z\). We use \(\ind\{\cdot\}\) for the indicator of a condition and \(\ind_B(x)=\ind\{x\in B\}\).

The space of finite point measures on \(\mathsf E\) is
\begin{equation}
  \mathsf N_{\mathrm f}(\mathsf E)
  \triangleq
  \left\{\sum_{i=1}^{k}\delta_{z_i}:
    k\in\{0,1,2,\ldots\},\ z_i\in\mathsf E\right\}.
  \label{eq:finite_point_measure_space}
\end{equation}
Equip \(\mathsf N_{\mathrm f}(\mathsf E)\) with the \(\sigma\)-algebra
\begin{equation}
  \B\!\left(\mathsf N_{\mathrm f}(\mathsf E)\right)
  \triangleq
  \sigma\left\{\nu\mapsto\nu(B):B\in\B(\mathsf E)\right\}.
  \label{eq:point_measure_sigma_algebra}
\end{equation}

\subsection{Crossing Operator}
\label{sec:crossing_operator}
\label{subsec:threshold_model_crossing_operator}
\label{subsec:stochastic_crossing_operator_kernel}

The crossing operator compares a signal path with threshold paths. Let \(\mathbf V\triangleq(V_n)_{n=1}^N\in\R^N\) be the nominal threshold vector. We identify each \(V_n\) with the constant path \(t\mapsto V_n\), and \(\mathbf V\) with the corresponding element of \(\mathsf X^N\). Equip \(\mathsf X^N\) with the product \(\sigma\)-algebra \(\B(\mathsf X)^{\otimes N}\), denoted by \(\B(\mathsf X^N)\).

Let \((\Omega_{\mathcal C},\mathcal F_{\mathcal C},\Prob_{\mathcal C})\) be an auxiliary complete probability space. The threshold noise \(\bm\eta=(\eta_n)_{n=1}^N\) is a measurable \(\mathsf X^N\)-valued random element satisfying \(\Exp_{\mathcal C}\bigl[\max_{n\in[N]}\lVert\eta_n\rVert_\infty\bigr]<\infty\). For \(\omega_{\mathcal C}\in\Omega_{\mathcal C}\), define the vector of threshold paths \(\widetilde{\mathbf V}=(\widetilde V_n)_{n=1}^N\) by
\begin{equation}
  \widetilde{\mathbf V}(\omega_{\mathcal C})
  \triangleq\mathbf V+\bm\eta(\omega_{\mathcal C}).
  \label{eq:instantaneous_thresholds}
\end{equation}

For a signal path \(x\in\mathsf X\) and threshold paths \(\mathbf v=(v_n)_{n=1}^N\in\mathsf X^N\), define the comparison state at each threshold \(n\in[N]\) by
\begin{equation}
  q_n(t;x,\mathbf v)
  \triangleq\operatorname{sgn}\bigl(x(t)-v_n(t)\bigr)
  \in\{-1,0,+1\},
  \label{eq:comparison_state}
\end{equation}
where \(\operatorname{sgn}(0)=0\).

Equip \(\mathsf X\times\mathsf X^N\) with the product \(\sigma\)-algebra \(\B(\mathsf X)\otimes\B(\mathsf X^N)\). Choose a Borel set \(\mathsf D_{\mathcal C}\subseteq\mathsf X\times\mathsf X^N\) as the admissible domain of the crossing operator. Throughout the operator construction, subsets of measurable spaces carry the inherited trace \(\sigma\)-algebras.

For each \((x,\mathbf v)\in\mathsf D_{\mathcal C}\) and each \(n\in[N]\), require a finite partition of \(\mathcal T\) into \(m_n\geq1\) subintervals such that, for \(j=1,\ldots,m_n\),
\begingroup
\begin{equation}
  \begin{gathered}[b]T_-=t_{n,0}<t_{n,1}<\cdots<t_{n,m_n}=T_+,\\
  q_n(\cdot;x,\mathbf v)\text{ is constant on }(t_{n,j-1},t_{n,j}).\end{gathered}
  \label{eq:finite_state_partition}
\end{equation}
\endgroup

At each interior time \(t\), the partition in \eqref{eq:finite_state_partition} gives one-sided states in \(\{-1,0,+1\}\) for every admissible pair and threshold. They are
\begin{equation}
  q_n(t^{\pm};x,\mathbf v)
  \triangleq
  \lim_{\varepsilon\downarrow0}q_n(t\pm\varepsilon;x,\mathbf v).
  \label{eq:one_sided_comparison_states}
\end{equation}
With \((x,\mathbf v)\) suppressed in \(q_n\), the set \(\mathcal R_n(x,\mathbf v)\) of analog crossing times at threshold \(n\) consists of all \(t\in(T_-,T_+)\) satisfying
\begin{equation}
  q_n(t^{-})\ne q_n(t) \quad\text{or}\quad q_n(t)\ne q_n(t^{+}).
  \label{eq:crossing_time_set}
\end{equation}
Assign to each \(t\in\mathcal R_n(x,\mathbf v)\) the direction mark
\begin{equation}
  \sigma_n(t;x,\mathbf v)
  \triangleq
  \frac{q_n(t^{+};x,\mathbf v)-q_n(t^{-};x,\mathbf v)}{2}.
  \label{eq:crossing_mark}
\end{equation}

Direction marks of magnitude \(1\) indicate transitions between opposite sides of the threshold; those of magnitude \(\tfrac12\) indicate transitions into or out of equality. Positive and negative signs indicate increases and decreases in the comparison state, respectively. A zero mark indicates an isolated contact, with the signal on the same side of the threshold immediately before and after.

\begin{definition}[Crossing operator]
\label{def:crossing_operator}
\label{def:stochastic_crossing_operator}
Define the crossing operator \(\mathcal C:\mathsf D_{\mathcal C}\to\mathsf N_{\mathrm f}(\mathsf E)\) by
\begin{equation}
  \mathcal C(x,\mathbf v)
  \triangleq
  \sum_{n=1}^{N}
  \sum_{t\in\mathcal R_n(x,\mathbf v)}
  \delta_{(t,n,\sigma_n(t;x,\mathbf v))}.
  \label{eq:crossing_operator}
\end{equation}

A signal path \(x\in\mathsf X\) is admissible, written \(x\in\mathsf X_{\mathcal C}\), if and only if
\begin{equation}
  \Prob_{\mathcal C}\bigl\{
    \omega_{\mathcal C}\in\Omega_{\mathcal C}:
    \bigl(x,\widetilde{\mathbf V}(\omega_{\mathcal C})\bigr)
    \in\mathsf D_{\mathcal C}
  \bigr\}
  =1.
  \label{eq:admissible_signal_set}
\end{equation}
For random thresholds, extend \(\mathcal C\) by the zero point measure outside \(\mathsf D_{\mathcal C}\). For each \(x\in\mathsf X_{\mathcal C}\), this extension affects only a null set of threshold realizations. The stochastic crossing operator on \(\mathsf X_{\mathcal C}\times\Omega_{\mathcal C}\) is
\begin{equation}
  \Phi_{\mathcal C}(x,\omega_{\mathcal C})
  \triangleq
  \mathcal C\!\left(
    x,\widetilde{\mathbf V}(\omega_{\mathcal C})
  \right).
  \label{eq:crossing_realization_map}
\end{equation}
\end{definition}

The crossing operator \(\mathcal C\) in \eqref{eq:crossing_operator} is measurable, the admissible signal set \(\mathsf X_{\mathcal C}\) is Borel, and \(\Phi_{\mathcal C}\) is jointly measurable.

\begin{proof}
Approximating \(t\) from above on finite time grids shows that \((x,t)\mapsto x(t)\) is jointly measurable. Thus the comparison states are jointly measurable, as are their one-sided limits, obtained by evaluating at \(t\pm\min\{t-T_-,T_+-t\}/(2m)\) and letting \(m\to\infty\).

With \((x,\mathbf v)\) ranging over \(\mathsf D_{\mathcal C}\) and \(n\) over \([N]\), the relation
\begin{equation*}
  \Gamma=\bigl\{
    \bigl((x,\mathbf v),(t,n,\sigma_n(t;x,\mathbf v))\bigr):
    t\in\mathcal R_n(x,\mathbf v)
  \bigr\}
\end{equation*}
is therefore Borel with finite sections. Since \(\mathsf D_{\mathcal C}\) and \(\mathsf E\) are standard Borel spaces, the Lusin--Novikov theorem~\cite[Thm.~18.10]{kechris1995classical} gives Borel partial maps \(z_j\) from \(\mathsf D_{\mathcal C}\) to \(\mathsf E\) whose disjoint graphs cover \(\Gamma\). Consequently, for every \(B\in\B(\mathsf E)\),
\begin{equation*}
  \mathcal C(x,\mathbf v)(B)=\sum_{j:\,(x,\mathbf v)\in\operatorname{dom}(z_j)}\ind_B\bigl(z_j(x,\mathbf v)\bigr)
\end{equation*}
is measurable. Equation~\eqref{eq:point_measure_sigma_algebra} therefore gives measurability of \(\mathcal C\).

Integrating \(\omega_{\mathcal C}\mapsto\ind_{\mathsf D_{\mathcal C}}\bigl(x,\widetilde{\mathbf V}(\omega_{\mathcal C})\bigr)\) with respect to \(\Prob_{\mathcal C}\) on \(\Omega_{\mathcal C}\) gives a Borel function of \(x\) whose level set at \(1\) is \(\mathsf X_{\mathcal C}\). Since \(\mathsf D_{\mathcal C}\) is Borel, the zero extension of \(\mathcal C\) is measurable. Composing it with \((x,\omega_{\mathcal C})\mapsto(x,\widetilde{\mathbf V}(\omega_{\mathcal C}))\) gives joint measurability of \(\Phi_{\mathcal C}\).
\end{proof}

At each sampling stage a Markov kernel \(K\) assigns an output distribution to each input \(a\), with \(a\mapsto K(a,B)\) measurable for every measurable output set \(B\)~\cite{kallenberg1997foundations}. For \(x\in\mathsf X_{\mathcal C}\) and measurable \(A\subseteq\mathsf N_{\mathrm f}(\mathsf E)\), joint measurability of \(\Phi_{\mathcal C}\) gives the crossing kernel
\begin{equation}
  K_{\mathcal C}(x,A)
  \triangleq
  \Prob_{\mathcal C}
  \left\{\Phi_{\mathcal C}(x,\cdot)\in A\right\}.
  \label{eq:crossing_kernel}
\end{equation}

\subsection{Timing Operator}
\label{sec:timing_operator}

The timing operator assigns continuous-valued or quantized time coordinates to analog crossings, producing timed crossings.

Equip \([N]\times\Sigma\) with the \(\sigma\)-algebra of all its subsets. Define the measurable projection \(\pi_{\mathrm m}:\mathsf E\to[N]\times\Sigma\) by \(\pi_{\mathrm m}(t,n,\sigma)\triangleq(n,\sigma)\). For \(\nu\in\mathsf N_{\mathrm f}(\mathsf E)\), the pushforward \((\pi_{\mathrm m})_{\#}\nu\) counts crossings by threshold index and direction mark and is given by
\begin{equation}
  \bigl((\pi_{\mathrm m})_{\#}\nu\bigr)(B)
  \triangleq
  \nu\bigl(\pi_{\mathrm m}^{-1}(B)\bigr),
  \qquad B\subseteq[N]\times\Sigma.
  \label{eq:channel_mark_count_measure}
\end{equation}

\begin{definition}[Timing operator]
\label{def:deterministic_timing_operator}
\label{def:stochastic_timing_operator}
A timing operator is a measurable map \(\mathcal Q:\mathsf N_{\mathrm f}(\mathsf E)\to\mathsf N_{\mathrm f}(\mathsf E)\) satisfying
\begin{equation}
  (\pi_{\mathrm m})_{\#}\mathcal Q(\nu)
  =
  (\pi_{\mathrm m})_{\#}\nu
  \label{eq:deterministic_channel_mark_conservation}
\end{equation}
for every \(\nu\in\mathsf N_{\mathrm f}(\mathsf E)\).

Let \((\Omega_{\mathcal Q},\mathcal F_{\mathcal Q},\Prob_{\mathcal Q})\) be a complete probability space. A stochastic timing operator is a jointly measurable map \(\Phi_{\mathcal Q}:\mathsf N_{\mathrm f}(\mathsf E)\times\Omega_{\mathcal Q}\to\mathsf N_{\mathrm f}(\mathsf E)\) such that, for \(\Prob_{\mathcal Q}\)-almost every \(\omega_{\mathcal Q}\), the map \(\Phi_{\mathcal Q}(\cdot,\omega_{\mathcal Q})\) is a timing operator.
\end{definition}

For \(\nu\in\mathsf N_{\mathrm f}(\mathsf E)\) and measurable \(A\subseteq\mathsf N_{\mathrm f}(\mathsf E)\), the induced timing kernel is
\begin{equation}
  K_{\mathcal Q}(\nu,A)
  \triangleq
  \Prob_{\mathcal Q}
  \{\Phi_{\mathcal Q}(\nu,\cdot)\in A\}.
  \label{eq:timing_realization_law}
\end{equation}

\subsection{Selection Operator}
\label{sec:selection_operator}

The selection operator retains or removes timed crossings, producing recorded crossings with unchanged times, threshold indices, and direction marks. For \(\nu',\nu\in\mathsf N_{\mathrm f}(\mathsf E)\), write \(\nu'\leq\nu\) if \(\nu'(B)\leq\nu(B)\) for every \(B\in\B(\mathsf E)\).
\begin{definition}[Selection operator]
\label{def:deterministic_selection_operator}
\label{def:stochastic_selection_operator}
A selection operator is a measurable map
\begin{equation}
  \mathcal H:\mathsf N_{\mathrm f}(\mathsf E)
  \to\mathsf N_{\mathrm f}(\mathsf E),
  \qquad \mathcal H(\nu)\leq\nu,
  \label{eq:deterministic_selection_operator}
\end{equation}
for every \(\nu\in\mathsf N_{\mathrm f}(\mathsf E)\).

Let \((\Omega_{\mathcal H},\mathcal F_{\mathcal H},\Prob_{\mathcal H})\) be a complete probability space. A stochastic selection operator is a jointly measurable map \(\Phi_{\mathcal H}:\mathsf N_{\mathrm f}(\mathsf E)\times\Omega_{\mathcal H}\to\mathsf N_{\mathrm f}(\mathsf E)\) such that, for \(\Prob_{\mathcal H}\)-almost every \(\omega_{\mathcal H}\), the map \(\Phi_{\mathcal H}(\cdot,\omega_{\mathcal H})\) is a selection operator.
\end{definition}

For \(\nu\in\mathsf N_{\mathrm f}(\mathsf E)\) and measurable \(A\subseteq\mathsf N_{\mathrm f}(\mathsf E)\), the induced selection kernel is
\begin{equation}
  K_{\mathcal H}(\nu,A)
  \triangleq
  \Prob_{\mathcal H}
  \{\Phi_{\mathcal H}(\nu,\cdot)\in A\}.
  \label{eq:selection_realization_law}
\end{equation}

\subsection{MT Sampling Operator}
\label{sec:mt_sampling_operator}

We compose the stochastic crossing, timing, and selection operators \(\Phi_{\mathcal C}\), \(\Phi_{\mathcal Q}\), and \(\Phi_{\mathcal H}\), whose internal randomness is mutually independent. On the product of their auxiliary probability spaces, denoted by \((\Omega_{\mathrm s},\mathcal F_{\mathrm s},\Prob_{\mathrm s})\), write \(\omega_{\mathrm s}=(\omega_{\mathcal C},\omega_{\mathcal Q},\omega_{\mathcal H})\).

The MT sampling operator acts on \(\mathsf X_{\mathcal C}\) at each realization \(\omega_{\mathrm s}\) by
\begin{equation}
  \Phi_{\mathrm s}(\cdot,\omega_{\mathrm s})
  \triangleq
  \Phi_{\mathcal H}(\cdot,\omega_{\mathcal H})
  \circ
  \Phi_{\mathcal Q}(\cdot,\omega_{\mathcal Q})
  \circ
  \Phi_{\mathcal C}(\cdot,\omega_{\mathcal C}).
  \label{eq:stochastic_mt_sampling_operator}
\end{equation}
The map \((x,\omega_{\mathrm s})\mapsto\Phi_{\mathrm s}(x,\omega_{\mathrm s})\) is jointly measurable on \(\mathsf X_{\mathcal C}\times\Omega_{\mathrm s}\).

For a kernel \(K\) and input distribution \(\mu\), write
\begin{equation}
  (\mu K)(B)\triangleq
  \int K(a,B)\,\mu(da)
  \label{eq:kernel_action}
\end{equation}
for the output distribution. For kernels with compatible input and output spaces, their composition is
\begin{equation}
  (K_1K_2)(a,B)\triangleq
  \int K_2(b,B)\,K_1(a,db).
  \label{eq:kernel_composition}
\end{equation}

By stage independence, the MT sampling kernel for \(x\in\mathsf X_{\mathcal C}\) and measurable \(A\subseteq\mathsf N_{\mathrm f}(\mathsf E)\) is
\begin{equation}
 \begin{aligned}[b]
 K_{\mathrm s}(x,A)&\triangleq\Prob_{\mathrm s}\{\Phi_{\mathrm s}(x,\cdot)\in A\}\\
 &=(K_{\mathcal C}K_{\mathcal Q}K_{\mathcal H})(x,A).
 \end{aligned}
 \label{eq:sampling_kernel_composition}
\end{equation}

Let \(Y\in\mathfrak Y\) have distribution \(P_Y\) with \(P_Y(\mathsf X_{\mathcal C})=1\), and let the sampler randomness be independent of \(Y\). On the product probability space \((\Omega\times\Omega_{\mathrm s},\mathcal F\otimes\mathcal F_{\mathrm s},\Prob\otimes\Prob_{\mathrm s})\), define the recorded crossings almost surely by
\begin{equation}
  \Xi(\omega,\omega_{\mathrm s})
  \triangleq
  \Phi_{\mathrm s}\bigl(Y(\omega),\omega_{\mathrm s}\bigr).
  \label{eq:recorded_crossings}
\end{equation}

By admissibility and Fubini's theorem~\cite[Thm.~A.13]{last2018lectures}, \(\bigl(Y(\omega),\widetilde{\mathbf V}(\omega_{\mathcal C})\bigr)\in\mathsf D_{\mathcal C}\) for \((\Prob\otimes\Prob_{\mathcal C})\)-almost every \((\omega,\omega_{\mathcal C})\in\Omega\times\Omega_{\mathcal C}\). With \(P_Y\) restricted to \(\mathsf X_{\mathcal C}\), the random finite point measure \(\Xi\) has distribution
\begin{equation}
  P_{\Xi}=P_YK_{\mathrm s}.
  \label{eq:recorded_crossing_law}
\end{equation}

\section{First Recorded Crossing Times}
\label{sec:crossing_properties}

\subsection{MT Sampler with Threshold Noise, Timing Jitter, and Dead Time}
\label{subsec:poisson_signal_sampler}

Let \(Y\in\mathfrak Y\) be as in Definition~\ref{def:signal_path_space}.

\begin{definition}[MT Sampler with Threshold Noise, Timing Jitter, and Dead Time]
\label{def:shared_noise_mt_sampler}
\textup{Crossing.} Fix the nominal thresholds \(\mathbf V\in\R^N\). Draw \(\bm\eta\) once per acquisition. Its components are independent with a common law on \(\mathsf X\) and satisfy \(\Exp_{\mathcal C}[\eta_n(t)]=0\) for all \(n\in[N]\), \(t\in\mathcal T\). Fix an admissible Borel domain \(\mathsf D_{\mathcal C}\) as in Definition~\ref{def:stochastic_crossing_operator} such that \(\bigl(Y(\omega),\widetilde{\mathbf V}(\omega_{\mathcal C})\bigr)\in\mathsf D_{\mathcal C}\) for \((\Prob\otimes\Prob_{\mathcal C})\)-almost every \((\omega,\omega_{\mathcal C})\in\Omega\times\Omega_{\mathcal C}\). The crossing operator is
\begin{equation}
  \Phi_{\mathcal C}(x,\omega_{\mathcal C})
  =\mathcal C\bigl(x,\mathbf V+\bm\eta(\omega_{\mathcal C})\bigr).
  \label{eq:shared_noise_crossing_operator}
\end{equation}

\textup{Timing.} Fix a quantization step \(\Delta>0\). On \(\Omega_{\mathcal Q}\), let \(J_{n,k}\) be the jitter applied before quantization to the \(k\)-th analog crossing at threshold \(n\). The array \(\bm J=(J_{n,k})_{n\in[N],\,k\ge1}\) has independent entries with common mean-zero law \(P_J\). The signal \(Y\), threshold noise \(\bm\eta\), and jitter array \(\bm J\) are mutually independent.

For \(\nu\in\mathsf N_{\mathrm f}(\mathsf E)\) and \(n\in[N]\), let
\(k_n(\nu)\triangleq\nu(\R\times\{n\}\times\Sigma)\) be the total mass of \(\nu\)
at threshold \(n\). List the atoms at this threshold with multiplicity as
\((s_{n,k},n,\sigma_{n,k})\), \(1\le k\le k_n(\nu)\), in order of time
and then by direction mark.
With
\(r_\Delta(u)=\Delta\lfloor u/\Delta+1/2\rfloor\), the timing operator
for a fixed array of jitter values \(\bm j=(j_{n,k})_{n\in[N],\,k\ge1}\) is
\begin{equation*}
  \mathcal Q_{\bm j}(\nu)
  =\sum_{n=1}^{N}\sum_{k=1}^{k_n(\nu)}
    \delta_{\left(r_\Delta(s_{n,k}+j_{n,k}),\,n,\,\sigma_{n,k}\right)}.
\end{equation*}
For each realization \(\omega_{\mathcal Q}\), set \(\Phi_{\mathcal Q}(\cdot,\omega_{\mathcal Q})=\mathcal Q_{\bm J(\omega_{\mathcal Q})}\).

\textup{Selection.} For a bounded open recording window \(W=(a,b)\subset\R\), the measure of timed crossings in \(W\) with direction mark \(\pm1\) is given, for \(B\in\B(\mathsf E)\), by
\[
 \nu_W^{\pm}(B)=\nu\bigl(B\cap(W\times[N]\times\{-1,+1\})\bigr).
\]
Fix a dead time \(\tau_{\mathrm{dead}}>0\) and let \(m=\nu_W^{\pm}(\mathsf E)\). If \(m>0\), list the atoms of \(\nu_W^{\pm}\) with multiplicity as \((t_i,n_i,\sigma_i)\), \(1\le i\le m\), in lexicographic order of \((n_i,t_i,\sigma_i)\). Start with \(\ell_1=1\). Given \(\ell_g\), let \(\ell_{g+1}\) be the smallest index \(i\) with \(\ell_g<i\le m\) satisfying
\begin{equation*}
  n_i\ne n_{\ell_g}\quad\text{or}\quad t_i-t_{\ell_g}\ge\tau_{\mathrm{dead}}.
\end{equation*}
Stop when no such index exists. The selection operator is
\begin{equation}
  \mathcal H_W(\nu)=\sum_g\delta_{(t_{\ell_g},n_{\ell_g},\sigma_{\ell_g})},
  \label{eq:shared_noise_mt_sampler}
\end{equation}
where the sum is empty when \(m=0\). For every \(\omega_{\mathcal H}\in\Omega_{\mathcal H}\), set \(\Phi_{\mathcal H}(\nu,\omega_{\mathcal H})=\mathcal H_W(\nu)\).
\end{definition}

With \(\omega_{\mathrm s}=(\omega_{\mathcal C},\omega_{\mathcal Q},\omega_{\mathcal H})\), the recorded crossings satisfy
\begin{equation}
  \Xi(\omega,\omega_{\mathrm s})
  =\mathcal H_W\!\left(\mathcal Q_{\bm J(\omega_{\mathcal Q})}\!\left(
    \mathcal C\bigl(Y(\omega),\widetilde{\mathbf V}(\omega_{\mathcal C})\bigr)\right)\right)
  \label{eq:shared_noise_poisson_recorded_crossings}
\end{equation}
for \((\Prob\otimes\Prob_{\mathrm s})\)-almost every \((\omega,\omega_{\mathrm s})\in\Omega\times\Omega_{\mathrm s}\). Supplemental Sec.~\ref{si-sec:supp_signal_sampler_measurability} verifies the operator conditions of Sec.~\ref{sec:sampling_operators}. Unqualified probabilities and expectations use the product law in \eqref{eq:recorded_crossings}, with independent factors for auxiliary draws.

\subsection{Joint and Marginal Distributions}
\label{subsec:poisson_exact_crossing_distributions}

\begin{definition}[First recorded crossing time]
\label{def:selected_crossing_time}
Let \(W=(a,b)\) be the recording window in Definition~\ref{def:shared_noise_mt_sampler}. For \(n\in[N]\) and \(\nu\in\mathsf N_{\mathrm f}(\mathsf E)\), define
\begin{equation}
  \tau_{n,W}(\nu)
  \triangleq\inf\{t\in W:
    \nu(\{t\}\times\{n\}\times\Sigma)>0\},
  \label{eq:poisson_crossing_selector}
\end{equation}
with \(\inf\varnothing=+\infty\). Thus \(\tau_{n,W}\) takes values in \(W\cup\{+\infty\}\). For the recorded crossing measure \(\Xi\) in \eqref{eq:shared_noise_poisson_recorded_crossings}, write \(\tau_n=\tau_{n,W}(\Xi)\).
\end{definition}

Selection preserves the earliest timed crossing in \(W\) at threshold \(n\) with direction mark \(\pm1\). Thus \(\tau_{n,W}(\mathcal H_W(\nu))=\tau_{n,W}(\nu_W^{\pm})\) for every timed crossing measure \(\nu\). Section~\ref{si-sec:supp_signal_sampler_measurability} of the Supplemental Material proves measurability of \(\tau_{n,W}\).

\subsubsection{Joint CDF of First Recorded Crossing Times}
\label{subsubsec:joint_distribution_across_thresholds}
For a nonempty set of threshold indices \(\mathcal I\subseteq[N]\) and cutoff vector \(\bm t=(t_n)_{n\in\mathcal I}\in\R^{\mathcal I}\), the joint CDF under the independent jitter of Definition~\ref{def:shared_noise_mt_sampler} is \(F_{\mathcal I}(\bm t)=\Prob\{\tau_n\le t_n\text{ for all }n\in\mathcal I\}\). For shared jitter, replace every entry in row \(n\) of the jitter array by \(J_n=J_{n,1}\) before applying the same quantization and selection rules. The resulting joint CDF is \(F_{\mathcal I,\mathrm{sh}}\). The \(J_n\) have law \(P_J\) and are independent across thresholds and independent of \((Y,\bm\eta)\).

For \(t\in\R\), the quantization cells of the consecutive grid points in \(W\cap(-\infty,t]\) form an empty set or a half-open interval \(G_t=r_\Delta^{-1}(W\cap(-\infty,t])\). Under a shift \(j\), an analog crossing has its quantized time in \(W\cap(-\infty,t]\) exactly when its analog time lies in
\begin{equation}
 A_t(j)=\{s\in(T_-,T_+):s+j\in G_t\}.
 \label{eq:timed_crossing_window}
\end{equation}
At threshold \(n\), let \(C_n\) be the finite point measure of analog crossing times with direction mark \(+1\) or \(-1\). For \((\omega,\omega_{\mathcal C})\in\Omega\times\Omega_{\mathcal C}\) and any Borel set \(B\subseteq\R\),
\[
 C_n(\omega,\omega_{\mathcal C})(B)
 \!=\!\mathcal C\bigl(Y(\omega),\widetilde{\mathbf V}(\omega_{\mathcal C})\bigr)\bigl(B\!\times\!\{n\}\!\times\!\{-1,+1\}\bigr).
\]
Under independent jitter, an atom of \(C_n\) uses \(J_{n,k}\) if its rank among all analog crossings at threshold \(n\) is \(k\), as in Definition~\ref{def:shared_noise_mt_sampler}. Under shared jitter, \(C_n(A_{t_n}(J_n))\) counts the atoms of the realized measure \(C_n\) in the set \(A_{t_n}(J_n)\) determined by the common shift for that realization. The identity following Definition~\ref{def:selected_crossing_time} then gives
\[
 F_{\mathcal I,\mathrm{sh}}(\bm t)
 =\Prob\!\left(\bigcap_{n\in\mathcal I}
       \{C_n(A_{t_n}(J_n))>0\}\right).
\]

To express this event through the path infimum and supremum, set \(D_n(s)=\eta_n(s)-\eta_n(T_-)\) for \(s\in\mathcal T\). Comparing \(Y\) with \(\widetilde V_n\) is equivalent to comparing \(Y-D_n\) with the constant level \(\widetilde V_n(T_-)\). Write \(\operatorname{lo}_n(t_n)\) and \(\operatorname{hi}_n(t_n)\) for the infimum and supremum of \(Y-D_n\) on the interior of \(A_{t_n}(J_n)\), taking both to be zero when the interior is empty.

These values are fixed given \(Y,D_n,J_n\). Since the threshold noises have the same law, choose one regular conditional distribution kernel \(\kappa\) from \(\mathsf X\) to \(\R\)~\cite[Thm.~5.3]{kallenberg1997foundations} such that, for each \(n\) and Borel \(B\subseteq\R\),
\[
 \kappa(D_n,B)=\Prob\{\eta_n(T_-)\in B\mid D_n\}\quad\text{a.s.}
\]

To compare independent and shared jitter, count pairs of crossings whose inclusion changes between two draws of the shift. Fix \(n,t\) and write \(J=J_n\). For an atom \(s\) of \(C_n\), \(w_t(s)=P_J\{j:s\in A_t(j)\}\) is its inclusion probability. Let \(J'\sim P_J\) be independent of \((J,Y,\bm\eta)\). For the symmetric difference \(A_t^{\triangle}=A_t(J)\mathbin{\triangle}A_t(J')\), define
\[
 B_n(t)=\Exp\min\!\left\{1,\tfrac12\Exp_{J,J'}
             \tbinom{C_n(A_t^{\triangle})}{2}\right\}.
\]
The inner expectation averages over \(J,J'\) with \(C_n\) fixed. A subscript \(J\) on an expectation or covariance denotes averaging over \(J\) alone.

\Needspace{14\baselineskip}
\begin{theorem}[Joint CDF under shared jitter and error bound for independent jitter]
\label{thm:joint_recorded_cdf}
Fix \(\bm t\in\R^{\mathcal I}\). Let \(Y\in\mathfrak Y\) be as in Definition~\ref{def:signal_path_space} and sample this \(Y\) as in Definition~\ref{def:shared_noise_mt_sampler}.

\textup{(i) Shared-jitter formula.} Suppose that, almost surely, for every \(n\in\mathcal I\), \(Y\) equals \(\widetilde V_n\) on no interval of positive length and \(C_n(\partial A_{t_n}(J_n))=0\). Then
\begin{align}
 &F_{\mathcal I,\mathrm{sh}}(\bm t)\notag\\
 &=\Prob\!\left(\bigcap_{n\in\mathcal I}
       \bigl\{\operatorname{lo}_n(t_n)<\widetilde V_n(T_-)<\operatorname{hi}_n(t_n)\bigr\}\right)
 \label{eq:joint_range}\\
 &=\Exp\biggl[\prod_{n\in\mathcal I} \kappa\Bigl(D_n,
       \bigl(\operatorname{lo}_n(t_n)-V_n,\operatorname{hi}_n(t_n)-V_n\bigr)\Bigr)\biggr].
 \label{eq:joint_conditional_product}
\end{align}

\textup{(ii) Comparison with independent jitter.} Under Definitions~\ref{def:signal_path_space} and~\ref{def:shared_noise_mt_sampler} alone,
\begin{equation}
 \bigl|F_{\mathcal I}(\bm t)-F_{\mathcal I,\mathrm{sh}}(\bm t)\bigr|
 \le \min\!\left\{1,\sum_{n\in\mathcal I}B_n(t_n)\right\}.
 \label{eq:joint_independent_shared_bound}
\end{equation}
The two CDFs agree at \(\bm t\) if, almost surely, at each \(n\in\mathcal I\) some atom \(s\) of \(C_n\) has \(w_{t_n}(s)=1\) or at most one atom has \(0<w_{t_n}(s)<1\).
\end{theorem}

\begin{proof}
Fix \(n\in\mathcal I\). The boundary condition in (i) allows us to test for a crossing on the interior of \(A_{t_n}(J_n)\). Since there is no interval of equality, the comparison state on each open subinterval of the finite partition in \eqref{eq:finite_state_partition} is \(+1\) or \(-1\). By right continuity of \(Y-\widetilde V_n\), each strict sign in the interior persists on an interval to its right. Both signs therefore occur in the interior exactly when adjacent subintervals of opposite sign meet there, giving a crossing with direction mark \(+1\) or \(-1\). Equivalently, \(\widetilde V_n(T_-)\) lies strictly between the path infimum and supremum. Taking the intersection over \(n\) gives \eqref{eq:joint_range}.

The path infima and suprema are measurable by right continuity and evaluation at rational times. Given \(Y\), the \(D_n\), and the \(J_n\), the initial values \(\eta_n(T_-)\) have conditional joint law \(\bigotimes_{n\in\mathcal I}\kappa(D_n,\cdot)\) by independence. Integrating over these values gives \eqref{eq:joint_conditional_product}.

For the error bound, condition on all analog crossings. At a fixed threshold \(n\), enumerate the analog crossings counted by \(C_n\) as \(s_1,\ldots,s_k\), with \(k=C_n(\R)\), and put \(t=t_n\) and \(f_i(j)=\ind\{s_i\notin A_t(j)\}\). The conditional probabilities of no recorded crossing by \(t\) are \(\prod_i\Exp_J[f_i(J)]\) for independent timing jitter and \(\Exp_J[\prod_i f_i(J)]\) for shared timing jitter, with empty product one. If some \(w_t(s_i)=1\), both conditional probabilities vanish. Otherwise each factor with \(w_t(s_i)=0\) equals one almost surely, so a difference requires at least two indices with \(0<w_t(s_i)<1\).

Assign a common \(J\) successively to the first \(i\) crossings, leaving the remaining jitter values independent. The change from \(i-1\) to \(i\) in the conditional probability of no recorded crossing is
\[
 \operatorname{Cov}_J\!\left(f_i(J),\prod_{r<i}f_r(J)\right)
                    \prod_{r>i}\Exp_J[f_r(J)].
\]
Write the covariance as half the expected product of the differences obtained by replacing \(J\) with \(J'\) in both arguments. Since \(|f_i(J)-f_i(J')|=\ind_{A_t^{\triangle}}(s_i)\), telescoping the product over \(r<i\) bounds the absolute change by \(\tfrac12\sum_{r<i}\Exp_{J,J'}[\ind_{A_t^{\triangle}}(s_i)\ind_{A_t^{\triangle}}(s_r)]\). Summing over \(i\) counts each unordered pair once. The absolute difference at threshold \(n\) is also at most one, so its expectation is bounded by \(B_n(t_n)\).

Under either choice of timing jitter the conditional joint probability factors across thresholds, since the jitter is independent across thresholds and dead time is applied separately. Telescoping the two products bounds their absolute difference by the sum of the single-threshold absolute differences. Averaging and using the bound of one gives \eqref{eq:joint_independent_shared_bound}. Under the stated condition on \(w_{t_n}\), every pair of corresponding factors agrees.
\end{proof}

\subsubsection{Marginal CDF of the First Recorded Crossing Time}
\label{subsubsec:single_threshold_distribution}
Take \(\mathcal I=\{n\}\) and \(t_n=t\in\R\). Write \(\tau=\tau_n\), \(F(t)=F_{\{n\}}(t)=\Prob\{\tau\le t\}\), and \(F_{\mathrm{sh}}(t)=F_{\{n\},\mathrm{sh}}(t)=\Prob\{C_n(A_t(J))>0\}\), with \(J=J_n\). We suppress \(n\) on \(\operatorname{lo}_n,\operatorname{hi}_n\).

\begin{corollary}[Marginal CDF and Fourier inversion]
\label{cor:range_crossing_fourier}
Under the path and boundary conditions of Theorem~\ref{thm:joint_recorded_cdf}(i) at \(n,t\),
\begin{align}
 F_{\mathrm{sh}}(t)&=\Prob\{\operatorname{lo}(t)<\widetilde V_n(T_-)<\operatorname{hi}(t)\}\notag\\
 &=\Exp\!\left[\kappa\big(D_n,(\operatorname{lo}(t)-V_n,\operatorname{hi}(t)-V_n)\big)\right].
 \label{eq:range_crossing_cdf}
\end{align}
For the Fourier representation, set \(\mathrm i=\sqrt{-1}\) and define, for \(\xi\in\R\),
\[
 \Psi_t(\xi)=\Exp\!\left[e^{-\mathrm i\xi\eta_n(T_-)}
       (e^{\mathrm i\xi \operatorname{hi}(t)}-e^{\mathrm i\xi \operatorname{lo}(t)})\right].
\]
Fourier inversion with the endpoint correction gives
\begin{align}
 F_{\mathrm{sh}}(t)
 &=\frac1\pi\lim_{R\to\infty}\int_0^R
 \frac{\operatorname{Im}[e^{-\mathrm i\xi V_n}\Psi_t(\xi)]}{\xi}\,d\xi\notag\\
 &\quad-\tfrac12\Prob\!\left\{\begin{gathered}
       \operatorname{lo}(t)<\operatorname{hi}(t),\\[-2pt]
       \widetilde V_n(T_-)\in\{\operatorname{lo}(t),\operatorname{hi}(t)\}
     \end{gathered}\right\}.
 \label{eq:range_recorded_crossing_fourier}
\end{align}
If the agreement criterion in Theorem~\ref{thm:joint_recorded_cdf}(ii) also holds at \(n,t\), both formulas give \(F(t)\). Under Definitions~\ref{def:signal_path_space} and~\ref{def:shared_noise_mt_sampler} alone,
\begin{equation}
 |F(t)-F_{\mathrm{sh}}(t)|\le B_n(t).
 \label{eq:independent_shared_crossing_bound}
\end{equation}
\end{corollary}

\begin{proof}
Taking \(\mathcal I=\{n\}\) in Theorem~\ref{thm:joint_recorded_cdf} gives \eqref{eq:range_crossing_cdf}, \eqref{eq:independent_shared_crossing_bound}, and the condition for agreement.

Set \(x_{\mathrm{hi}}=\operatorname{hi}(t)-\widetilde V_n(T_-)\) and \(x_{\mathrm{lo}}=\operatorname{lo}(t)-\widetilde V_n(T_-)\). The numerator in \eqref{eq:range_recorded_crossing_fourier} is \(\Exp[\sin(\xi x_{\mathrm{hi}})-\sin(\xi x_{\mathrm{lo}})]\). For \(\xi\ne0\), the absolute value of the sine difference divided by \(|\xi|\) is at most \(x_{\mathrm{hi}}-x_{\mathrm{lo}}=\operatorname{hi}(t)-\operatorname{lo}(t)\le 2\lVert Y-D_n\rVert_\infty\). This bound is integrable since \(\lVert D_n\rVert_\infty\le 2\lVert\eta_n\rVert_\infty\) and both \(\lVert Y\rVert_\infty\) and \(\lVert\eta_n\rVert_\infty\) are integrable by Definition~\ref{def:signal_path_space} and Sec.~\ref{sec:crossing_operator}, respectively. Fubini applies on \([0,R]\), and dominated convergence gives the zero-frequency limit \(\Exp[x_{\mathrm{hi}}-x_{\mathrm{lo}}]\).

The sine integral \(\int_0^u(\sin v)/v\,dv\) is uniformly bounded and tends to \(\pm\pi/2\) as \(u\to\pm\infty\). By dominated convergence, the truncated inversion integral~\cite{gil1951note} divided by \(\pi\) therefore tends to \(\tfrac12\Exp[\operatorname{sgn}(x_{\mathrm{hi}})-\operatorname{sgn}(x_{\mathrm{lo}})]\). The pathwise contribution is one on \(x_{\mathrm{lo}}<0<x_{\mathrm{hi}}\), one half when \(x_{\mathrm{lo}}<x_{\mathrm{hi}}\) and one endpoint is zero, and zero otherwise, including \(x_{\mathrm{lo}}=x_{\mathrm{hi}}=0\). Subtracting this endpoint contribution gives \eqref{eq:range_recorded_crossing_fourier}. If \(\kappa(D_n,\cdot)\) is almost surely atomless, the correction vanishes by conditioning on \(Y,D_n,J\).
\end{proof}

Under additional path regularity, Supplemental Theorem~\ref{si-thm:supp_independent_interval_cdf} and Corollary~\ref{si-cor:supp_independent_joint_cdf} express the exact marginal and joint CDFs under independent per-crossing jitter in terms of random unions of amplitude intervals.

\Needspace{9\baselineskip}
\subsection{Conditional Standard Deviation of the First Recorded Crossing Time}
\label{sec:poisson_crossing_sd}

The recording probability is \(p=F(b)=\Prob\{\tau<\infty\}\), leaving mass \(1-p\) at \(+\infty\). For \(p>0\), the conditional CDF given a recording at threshold \(n\) is \(F(t)/p\) at finite \(t\). The corresponding SD \(S\triangleq\sqrt{\operatorname{Var}(\tau\mid\tau<\infty)}\) measures single-detector timing precision and satisfies
\begin{IEEEeqnarray}{c}
 S=\sqrt{\frac{2}{p}\int_a^b(b-t)F(t)\,dt
       -\left(\frac{1}{p}\int_a^b F(t)\,dt\right)^2}.
 \IEEEeqnarraynumspace
 \label{eq:exact_recorded_crossing_sd}
\end{IEEEeqnarray}

On \(\{\tau<\infty\}\), the time \(\tau\) lies in \((a,b)\). For \(k=1,2\), Tonelli's theorem~\cite[Thm.~A.13]{last2018lectures} gives
\begin{equation*}
 \Exp[(b-\tau)^k\mid\tau<\infty]
 =\frac{k}{p}\int_a^b(b-t)^{k-1}F(t)\,dt.
\end{equation*}
Since \(\operatorname{Var}(b-\tau\mid\tau<\infty)=\operatorname{Var}(\tau\mid\tau<\infty)\), subtracting the squared first moment from the second gives \eqref{eq:exact_recorded_crossing_sd}.

Both integrals in \eqref{eq:exact_recorded_crossing_sd} can be evaluated on the quantization grid, since \(F\) is constant between consecutive points of \(\{a,b\}\cup(W\cap\Delta\mathbb Z)\). Let \(\varepsilon\) and \(\varepsilon_b\) bound \(\sup_{t\in\R}|F(t)-F_{\mathrm{sh}}(t)|\) and \(|p-F_{\mathrm{sh}}(b)|\), respectively. For nonempty \(W\cap\Delta\mathbb Z\), \eqref{eq:independent_shared_crossing_bound} gives the choices \(\varepsilon=\max_{t\in W\cap\Delta\mathbb Z}B_n(t)\) and \(\varepsilon_b=B_n(b)\).

\begin{corollary}[Grid formula and error bounds for conditional SD]
\label{cor:shared_timing_jitter_sd_bound}
Assume \(p_{\mathrm{sh}}=F_{\mathrm{sh}}(b)>0\). Then \(W\cap\Delta\mathbb Z\) is nonempty. Write its points as \(g_1<\cdots<g_m\), so \(F_{\mathrm{sh}}(g_m)=p_{\mathrm{sh}}\).

\textup{(i) Shared-jitter formula.} Write \(S_{\mathrm{sh}}\) for the corresponding conditional SD. With \(c_j=F_{\mathrm{sh}}(g_j)/p_{\mathrm{sh}}\), \eqref{eq:exact_recorded_crossing_sd} gives
\begin{equation}
 S_{\mathrm{sh}}^2=\Delta^2\sum_{j=1}^{m-1}(1-c_j)
                  \left(c_j+2\sum_{i=1}^{j-1}c_i\right).
 \label{eq:shared_timing_jitter_sd_grid}
\end{equation}
At grid points satisfying the path and boundary conditions of Corollary~\ref{cor:range_crossing_fourier}, the numerator of \(c_j\) is given by \eqref{eq:range_crossing_cdf}.

\textup{(ii) Bounds for independent jitter.} Set \(d=g_m-g_1\) and \(\varepsilon_{\mathrm{var}}=d^2(\varepsilon+\varepsilon_b)/p_{\mathrm{sh}}\). Then \(p>0\) and
\[
 \sqrt{\max\{0,S_{\mathrm{sh}}^2-\varepsilon_{\mathrm{var}}\}}\le S
 \le\min\{d/2,\sqrt{S_{\mathrm{sh}}^2+\varepsilon_{\mathrm{var}}}\}.
\]
\end{corollary}

\begin{proof}
Under shared timing jitter, conditional on recording a crossing, the first recorded crossing time is \(g_1\) plus \(\Delta\) for each grid point it exceeds. For \(i\le j\) the corresponding tail indicators have covariance \(c_i(1-c_j)\), giving \eqref{eq:shared_timing_jitter_sd_grid}.

Given \(C_n\), the probability of recording a crossing is positive under both independent and shared timing jitter exactly when \(w_b(s)>0\) for some atom \(s\) of \(C_n\). Thus \(p_{\mathrm{sh}}>0\) implies \(p>0\). The identity
\[
 p_{\mathrm{sh}}\left(\frac Fp-\frac{F_{\mathrm{sh}}}{p_{\mathrm{sh}}}\right)
 =F-F_{\mathrm{sh}}-(p-p_{\mathrm{sh}})\frac Fp
\]
bounds the absolute difference of the conditional CDFs by \((\varepsilon+\varepsilon_b)/p_{\mathrm{sh}}\).

For any \(x\in[g_1,g_m]\), integration by parts bounds the absolute difference between the integrals of \(t\mapsto(t-x)^2\) against the two conditional laws by
\[
 \frac{\varepsilon+\varepsilon_b}{p_{\mathrm{sh}}}
 \int_{g_1}^{g_m} 2|t-x|\,dt
 \le\varepsilon_{\mathrm{var}}.
\]
Each conditional variance is the minimum over \(x\in[g_1,g_m]\) of the corresponding second moment about \(x\). Since the bound holds for every such \(x\), it also bounds \(|S^2-S_{\mathrm{sh}}^2|\). Taking \(x=(g_1+g_m)/2\) gives \(S,S_{\mathrm{sh}}\le d/2\), proving the stated interval.
\end{proof}

Consider two independent detector channels, each with positive recording probability at threshold \(n\). Given that both record, Supplemental Proposition~\ref{si-prop:unbiased_time_difference} states when the difference of their first recorded crossing times is unbiased for the difference of their finite waveform onsets (Supplemental Definition~\ref{si-def:waveform_onset}). The conditional variance of this time difference equals the sum of the single-channel conditional variances (Supplemental Proposition~\ref{si-prop:time_difference_variance}), giving the SD-based coincidence time resolution (Supplemental Definition~\ref{si-def:ctr}).

\subsection{Photon-Arrival Signals and the Biexponential--Gaussian Model}
\label{subsec:biexponential_gaussian_model}
Each detected photon contributes a response shifted to its arrival time~\cite{choong2009timing}.

\begin{definition}[Signal generated by photon arrivals]
\label{def:poisson_photon_signal}
Let the parameter set \(\Theta\) be a nonempty Borel subset of a Euclidean space of finite positive dimension. For each \(\bm\theta\in\Theta\), let \(\Pi_{\bm\theta}=\sum_i\delta_{U_i(\bm\theta)}\) be a locally finite point process of photon arrivals on \(\R\), defined on \((\Omega,\mathcal F,\Prob)\). The arrival times \(U_i(\bm\theta)\), abbreviated \(U_i\), are listed with multiplicity in a fixed measurable order. Local finiteness means that, almost surely, every bounded interval contains finitely many arrivals.

On \(D(\R;\R)\), the space of right-continuous paths with finite left limits, use the coordinate \(\sigma\)-algebra. Fix a probability law on this space, common to all \(\bm\theta\in\Theta\). Conditional on the arrival times \((U_i)_i\), the corresponding single-photon responses \(H_i\) are independent and identically distributed with the specified law. The sample paths of each \(H_i\) are almost surely bounded on \(\R\).

For each \(\bm\theta\in\Theta\) and each compact interval \(L\subset\R\), assume that
\begin{equation*}
  \Exp\!\left[\sum_i\sup_{t\in L}|H_i(t-U_i)|\right]<\infty.
\end{equation*}
Define the signal by
\begin{equation}
  Y_{\bm\theta}(t)\triangleq\sum_i H_i(t-U_i),
  \qquad t\in\R.
  \label{eq:poisson_photon_signal}
\end{equation}
\end{definition}

For fixed \(\bm\theta\), \(Y=Y_{\bm\theta}|_{\mathcal T}\in\mathfrak Y\) (Supplemental Sec.~\ref{si-sec:supp_signal_sampler_measurability}). Supplemental Proposition~\ref{si-prop:waveform_arrival_onsets} gives conditions under which its waveform onset equals the photon-arrival onset \(t_0\) of Supplemental Definition~\ref{si-def:photon_arrival_onset}.

We model detected arrivals with the biexponential shape of the scintillation emission rate~\cite[Eq.~(2)]{choong2009timing} and the mean single-photon response with the fast rise and slower decay of a silicon photomultiplier~\cite{van2010comprehensive,acerbi2019understanding}.

\begin{definition}[Biexponential--Gaussian model]
\label{def:biexponential_gaussian_model}
In the photon-arrival model of Definition~\ref{def:poisson_photon_signal}, let \(\bm\theta=(\Lambda,t_0,\alpha,\beta)\), with \(\Lambda>0\), \(t_0\in\R\), and \(0<\alpha<\beta\). Let \(\Pi_{\bm\theta}\) be the nonhomogeneous Poisson process of detected photon arrivals. Its arrival intensity \(\lambda_{\bm\theta}\), with location parameter \(t_0\), and the mean single-photon response \(h\) are
\begin{gather}
 \lambda_{\bm\theta}(t)=\frac{\Lambda}{\beta-\alpha}
 \bigl(e^{-(t-t_0)/\beta}-e^{-(t-t_0)/\alpha}\bigr)\ind\{t\geq t_0\},\notag\\
 h(u)=A(e^{-u/\tau_{\mathrm d}}-e^{-u/\tau_{\mathrm r}})\ind\{u\geq0\},
 \label{eq:scintillation_example_model}
\end{gather}
where \(A>0\) sets the amplitude scale and \(0<\tau_{\mathrm r}<\tau_{\mathrm d}\). Let \(N_{\mathrm{ph}}=\Pi_{\bm\theta}(\R)\) be the total detected photon count and index these arrivals in time order.

Let \((H_i)_{i\ge1}\) be independent single-photon responses. Their fluctuations \(H_i-h\) are centered Gaussian processes with continuous paths. For delays \(u,v\in\R\), the responses have common covariance
\[
 \operatorname{Cov}\bigl(H_i(u),H_i(v)\bigr)
 =s_h^2 h(u)h(v)\exp\!\left(-\frac{(u-v)^2}{2\ell_h^2}\right).
\]
Here \(s_h>0\) is the relative response SD at positive delays and \(\ell_h>0\) is the correlation time scale.

In the sampler of Definition~\ref{def:shared_noise_mt_sampler}, take the threshold noise processes \(\eta_n\) to be independent centered Gaussian processes with continuous paths, standard deviation \(s_\eta>0\), and correlation time scale \(\ell_\eta>0\). For \(s,t\in\mathcal T\), their common covariance is
\[
 \operatorname{Cov}\bigl(\eta_n(s),\eta_n(t)\bigr)
 =s_\eta^2\exp\!\left(-\frac{(s-t)^2}{2\ell_\eta^2}\right).
\]
The timing jitter variables \(J_{n,k}\) are independent with common law \(\mathcal N(0,s_J^2)\), with \(s_J>0\). The random elements \(\Pi_{\bm\theta}\), \((H_i)_{i\ge1}\), \(\bm\eta\), and \(\bm J\) are mutually independent. The response and sampler parameters are fixed across \(\bm\theta\).
\end{definition}

For \(T_-\le t_0<T_+\), both onsets equal \(t_0\) (Supplemental Corollary~\ref{si-cor:bg_waveform_onset}). Supplemental Sec.~\ref{si-sec:supp_signal_sampler_measurability} verifies integrability and admissibility.

The ordered arrival vector \(\bm U=(U_i)_{i=1}^{N_{\mathrm{ph}}}\) has length \(N_{\mathrm{ph}}\sim\operatorname{Pois}(\Lambda)\) and is empty at count zero. Given \(N_{\mathrm{ph}}=k\), its entries are distributed as the order statistics of \(k\) independent draws with density \(\lambda_{\bm\theta}/\Lambda\)~\cite[Prop.~3.8]{last2018lectures}. Given \(\bm U\), the time shifts in the sum of independent Gaussian responses are fixed, so \(Y\) is conditionally Gaussian. For \(s,v\in\mathcal T\), its conditional mean at \(s\) is \(\sum_{i=1}^{N_{\mathrm{ph}}}h(s-U_i)\) and its conditional covariance is
\begin{multline}
 \operatorname{Cov}(Y(s),Y(v)\mid\bm U)\\
 =s_h^2 e^{-(s-v)^2/(2\ell_h^2)}
   \sum_{i=1}^{N_{\mathrm{ph}}} h(s-U_i)h(v-U_i).
 \label{eq:bg_signal_cov}
\end{multline}

To compare the signal with a level that is constant in time, let \(Z_n\) be \(\int_{\mathcal T}\eta_n(v)\,dv\) divided by its standard deviation. Define the coefficient \(a_\eta\) and the residual \(\eta_n^\perp\) by
\[
 a_\eta(s)=\operatorname{Cov}(\eta_n(s),Z_n),\quad
 \eta_n^\perp(s)=\eta_n(s)-a_\eta(s)Z_n .
\]
By Gaussian regression, the residual \(\eta_n^\perp\) is independent of the standard normal variable \(Z_n\). The threshold noise covariance in Definition~\ref{def:biexponential_gaussian_model} is strictly positive, so \(a_\eta(s)>0\). Thus, with
\[
 L_n(s)=\frac{Y(s)-V_n-\eta_n^\perp(s)}{a_\eta(s)},
\]
the signal exceeds the threshold value \(\widetilde V_n(s)\) exactly when \(L_n(s)>Z_n\).

\begin{corollary}[CDFs and conditional SD under shared jitter]
\label{cor:bg_joint}
In the model of Definition~\ref{def:biexponential_gaussian_model}, apply the same timing jitter \(J_n=J_{n,1}\) to all analog crossings at threshold \(n\). The \(J_n\) are independent draws from \(\mathcal N(0,s_J^2)\) and are jointly independent of the signal and threshold noise. Write \(\Phi\) for the standard normal CDF.

\textup{(i) Joint CDF.} For nonempty \(\mathcal I\subseteq[N]\) and \(\bm t\in\R^{\mathcal I}\),
\begin{align}
 F_{\mathcal I,\mathrm{sh}}(\bm t)
 &=\Exp\biggl[\prod_{n\in\mathcal I}\biggl\{
 \Phi\!\left(\sup_{s\in A_{t_n}(J_n)^\circ}L_n(s)\right)\notag\\[-2pt]
 &\hspace{3em}{}-\Phi\!\left(\inf_{s\in A_{t_n}(J_n)^\circ}L_n(s)\right)
 \biggr\}\biggr].
 \label{eq:bg_joint}
\end{align}
Here \(A_{t_n}(J_n)^\circ\) denotes the interval interior. The factor is zero when that interior is empty.

\textup{(ii) Marginal CDF.} Given \(L_n\), the \(n\)-th factor in \eqref{eq:bg_joint}, evaluated at \(t_n=t\) and \(J_n=j\), gives \(\Prob\{C_n(A_t(j))>0\mid L_n\}\).

Given \(C_n\), write its atoms as \(s_1<\cdots<s_k\), where \(k=C_n(\R)\), and set \(s_0=-\infty\). For nonempty \(G_t=[a_t,b_t)\),
\begin{multline}
 F_{\mathrm{sh}}(t)=\Exp\!\int_{\R}
 \Prob\{C_n(A_t(j))>0\mid L_n\}\,P_J(dj)\\
 =\Exp\sum_{i=1}^k P_J\bigl([a_t-s_i,\min\{b_t-s_i,a_t-s_{i-1}\})\bigr).
 \label{eq:bg_marginal_given_crossings}
\end{multline}
An empty \(G_t\) gives zero. For \(u<v\), \(P_J([u,v))=\Phi(v/s_J)-\Phi(u/s_J)\).

\textup{(iii) Conditional SD.} Suppose \(W\cap\Delta\mathbb Z\) is nonempty, and list its points as \(g_1<\cdots<g_m\). Equation~\eqref{eq:bg_marginal_given_crossings} gives the probability that the first recorded crossing time equals \(g_j\) as \(p_{j,\mathrm{sh}}=F_{\mathrm{sh}}(g_j)-F_{\mathrm{sh}}(g_j-\Delta)\). Then \(p_{\mathrm{sh}}=F_{\mathrm{sh}}(g_m)>0\), and the conditional SD is
\begin{equation}
 S_{\mathrm{sh}}=\sqrt{\frac{\sum_j p_{j,\mathrm{sh}}(g_j-\bar g_{\mathrm{sh}})^2}{p_{\mathrm{sh}}}},
 \quad
 \bar g_{\mathrm{sh}}=\frac{\sum_j g_j p_{j,\mathrm{sh}}}{p_{\mathrm{sh}}}.
 \label{eq:bg_shared_sd}
\end{equation}
\end{corollary}

\begin{proof}
The centered Gaussian process with covariance \(e^{-(u-v)^2/2}\) has an analytic version given by \(e^{-u^2/2}\sum_{k\ge0}\xi_k u^k/\sqrt{k!}\), where the \(\xi_k\) are independent standard normals. This series, obtained by restricting the Bargmann--Fock field to a coordinate axis~\cite[Eq.~(2.2)]{beliaev2023smooth}, converges locally uniformly almost surely. Rescaling its time argument by \(\ell_h\) and multiplying by \(s_h h\) gives the law of \(H_i-h\). The corresponding scales for \(\eta_n\) are \(\ell_\eta\) and \(s_\eta\). Since \(N_{\mathrm{ph}}\) is finite and \(a_\eta\) is positive and analytic, \(L_n\) is continuous and analytic between arrivals, with each piece extending analytically across its endpoints.

Condition on \(\bm U,Y\), all residuals \(\eta_n^\perp\), and all \(J_n\). By Gaussian regression and independence, the \(Z_n\) are conditionally independent standard normals. Each \(Z_n\) almost surely avoids the finite set of critical values of nonconstant pieces of \(L_n\)~\cite[Sec.~1.2]{krantz2002primer}, values of constant pieces, and values at piece endpoints and the endpoints of a nonempty \(A_{t_n}(J_n)\). On this event, a crossing is recorded by \(t_n\) exactly when \(A_{t_n}(J_n)^\circ\) is nonempty and \(Z_n\) lies strictly between the infimum and supremum of \(L_n\) on that interval. Integrating first over the independent levels \(Z_n\) and then over the conditioning variables gives \eqref{eq:bg_joint}.

Fix \(j\in\R\) and condition on \(\bm U,Y,\eta_n^\perp\). The level \(Z_n\) is still standard normal and almost surely avoids the exceptional values above for the fixed interval \(A_t(j)\). Integrating over \(Z_n\) therefore gives the factor in \eqref{eq:bg_joint} with \(t_n=t\) and \(J_n=j\). For fixed \(j\) this factor depends only on \(L_n\), so it is a version of \(\Prob\{C_n(A_t(j))>0\mid L_n\}\). Use this version for all \(j\), with value zero for empty interval interiors. This version is jointly measurable in \((L_n,j)\) because the infimum and supremum can be taken over the rational times in the shifted open interval. Integrating over \(P_J(dj)\) and averaging gives the first expression in \eqref{eq:bg_marginal_given_crossings}.

Given \(C_n\), \(s_i\) is the first atom of \(C_n\) shifted to or beyond \(a_t\) exactly when \(a_t-s_i\le J_n<a_t-s_{i-1}\). Its quantized time lies in \(W\cap(-\infty,t]\) if also \(J_n<b_t-s_i\). These disjoint intervals give the second expression after averaging over \(C_n\).

Given \(Y\), the threshold noise at two distinct interior times has a nonsingular Gaussian law, so opposite comparison signs, and hence an intervening crossing, occur with positive probability. The shared timing jitter \(J_n\) is Gaussian and independent of \(C_n\). Given \(C_n\) with first atom \(s_1\), \(J_n\) therefore has positive probability of falling in \([g_1-s_1-\Delta/2,g_1-s_1+\Delta/2)\), which sends \(s_1\) to \(g_1\) after quantization. Thus \(p_{\mathrm{sh}}>0\), and the conditional grid probabilities \(p_{j,\mathrm{sh}}/p_{\mathrm{sh}}\) give \eqref{eq:bg_shared_sd}.
\end{proof}

When shared timing jitter is used to approximate independent timing jitter, Theorem~\ref{thm:joint_recorded_cdf} and Corollaries~\ref{cor:range_crossing_fourier} and~\ref{cor:shared_timing_jitter_sd_bound} give the error bounds for \eqref{eq:bg_joint}--\eqref{eq:bg_shared_sd}, respectively.

\section{Numerical Experiments}
\label{sec:numerical_experiments}

\subsection{Model Parameters}
\label{subsec:model_parameters}
\begin{figure*}[!t]
 \centering
 \includegraphics[width=\textwidth]{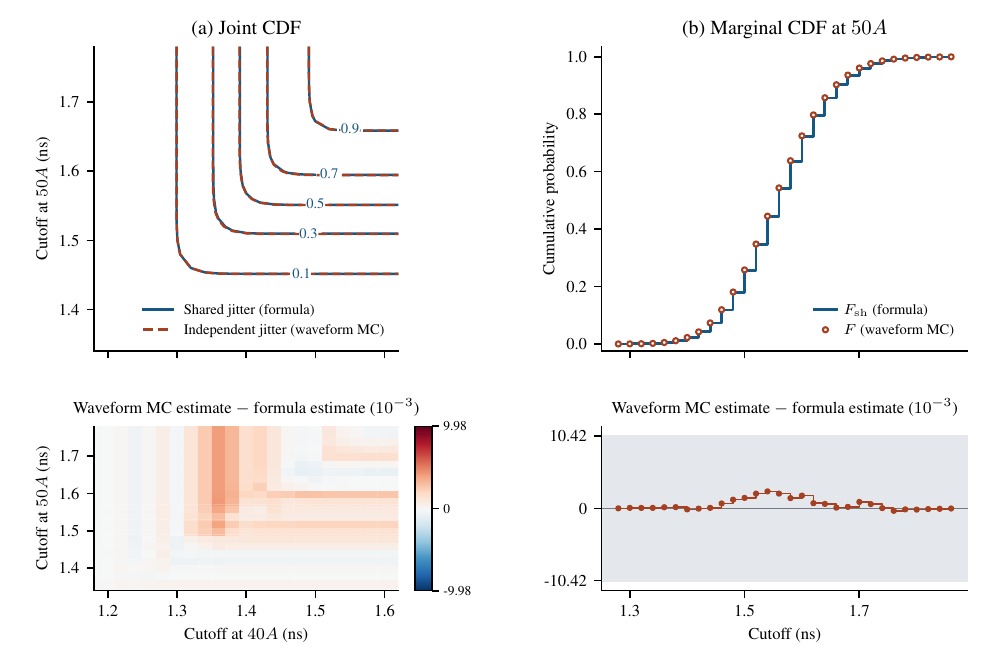}
 \caption{Shared-jitter formula estimates and independent-jitter waveform Monte Carlo (MC) estimates. (a) Joint CDF contours at \(40A\) and \(50A\), interpolated on the cutoff grid. (b) Rising portion of the marginal CDF at \(50A\). The lower panels subtract the formula estimate \eqref{eq:bg_joint} or \eqref{eq:bg_marginal_given_crossings} from the waveform MC estimate. Color limits and shading give the respective 95\% simultaneous bounds on Monte Carlo sampling error in these differences.}
 \label{fig:bg_cdfs}
\end{figure*}

We use the biexponential--Gaussian model of Definition~\ref{def:biexponential_gaussian_model} for a \(511\,\mathrm{keV}\) energy deposition in lutetium--yttrium oxyorthosilicate (LYSO) read out by a silicon photomultiplier (SiPM). From simulations of lutetium oxyorthosilicate we take the mean detected count \(\Lambda=4000\)~\cite{derenzo2015monte} and the decay constant \(\beta=40\,\mathrm{ns}\)~\cite{choong2009timing}. Single-component fits to the scintillation rise in two cerium-doped LYSO samples at \(511\,\mathrm{keV}\) gave \(68\pm13\) and \(77\pm14\,\mathrm{ps}\)~\cite[Table~4]{gundacker2018rise}. We choose \(\alpha=0.1\,\mathrm{ns}\) for the rise time constant of the arrival intensity and set \(t_0=0.2\,\mathrm{ns}\).

Amplitudes are in units of \(A\). A fired SiPM microcell has a rise time of about \(1\,\mathrm{ns}\) and decays over tens of nanoseconds~\cite{piatek2014physics}, motivating \(\tau_{\mathrm r}=1\,\mathrm{ns}\) and \(\tau_{\mathrm d}=20\,\mathrm{ns}\). Guided by the reported cell-to-cell gain spread~\cite{buzhan2001advanced}, we use relative response SD \(s_h=0.10\) at positive delays. We set the correlation time scale to \(\ell_h=1.0\,\mathrm{ns}\).

For the independent threshold noise processes \(\eta_n\), we set \(s_\eta=0.05A\), about \(6.2\%\) of the peak of \(h\). With \(\ell_\eta=0.50\,\mathrm{ns}\), the model threshold noise spectrum has a half-power frequency of approximately \(375\,\mathrm{MHz}\), comparable to the \(300\)--\(500\,\mathrm{MHz}\) bandwidth reported for a SiPM readout circuit~\cite{gundacker2023timing}.

The single-channel root-mean-square timing precision reported for input slopes above \(0.5\,\mathrm{V/ns}\) is about \(20\,\mathrm{ps}\)~\cite{palka2017multichannel}. This motivates our choice \(s_J=20\,\mathrm{ps}\) for the jitter SD. We use a uniform quantization step \(\Delta=20\,\mathrm{ps}\) and a per-threshold dead time \(\tau_{\mathrm{dead}}=5\,\mathrm{ns}\). For comparison, CERN's high-performance time-to-digital converter offers \(25\,\mathrm{ps}\) binning and a programmable channel dead time of approximately \(5\,\mathrm{ns}\)~\cite{christiansen2004hptdc}.

The observation and recording windows are \(\mathcal T=[0,25]\,\mathrm{ns}\) and \(W=(-0.3,25.3)\,\mathrm{ns}\). We evaluate the marginal CDF at \(50A\) and the joint CDF at \((V_1,V_2)=(40A,50A)\). For the conditional SD we test 33 nominal threshold settings over \(0\leq V_n/A\leq800\), spaced more finely below \(A\). The mean pulse first reaches the highest tested threshold, \(800A\), at \(15.70\,\mathrm{ns}\), within the observation window [Fig.~\ref{fig:bg_threshold_sd}(a)].

\subsection{Numerical Methods}
\label{subsec:experimental_setup_methods}

For each acquisition we draw photon arrivals with intensity \eqref{eq:scintillation_example_model}, then one Gaussian signal path with the conditional mean in Sec.~\ref{subsec:biexponential_gaussian_model} and covariance \eqref{eq:bg_signal_cov}. All thresholds use this path. We approximate each squared-exponential covariance factor by a sum of products of Gaussian basis functions on a lattice. In units of the corresponding correlation time scale, the lattice spacing is \(0.4\) and the cutoff radius is six. Independent Gaussian coefficients generate the two threshold noise paths for the joint CDF. We evaluate the paths on a \(10\,\mathrm{ps}\) mesh and at photon arrivals in \(\mathcal T\).

For the joint CDF \eqref{eq:bg_joint}, we project each threshold noise process onto its integral over \(\mathcal T\) and independently draw one shared shift \(J_n\) per threshold. We take the infimum and supremum of the linearly interpolated \(L_n\) on the shifted interval using the one-sided limits of \(L_n\) at its endpoints. For both the marginal formula and waveform Monte Carlo (MC), we locate analog crossings with mark \(\pm1\) by linear interpolation of \(Y-V_n-\eta_n\). Given \(C_n\), we integrate over the shared Gaussian shift to obtain the probability that the first recorded crossing time equals \(g_j\). For each crossing we retain the quantization cell containing its unshifted time and 11 cells on either side, omitting less than \(2\Phi(-10)\) of the jitter probability. Averaging these conditional probabilities over realizations of \(C_n\) gives the grid probabilities \(p_{j,\mathrm{sh}}\). Their cumulative sums give the marginal CDF in \eqref{eq:bg_marginal_given_crossings}, and \eqref{eq:bg_shared_sd} gives the conditional SD. Waveform MC draws independent jitters for the crossings, rounds the shifted times to the quantization grid, and takes the earliest time in \(W\) at each threshold.

For the main comparisons, the formula and waveform MC estimates use independent samples of \(100\,000\) acquisitions each. Within each acquisition, the SD scan reuses the threshold noise path across settings. Waveform MC also reuses the jitter draw for the \(k\)-th crossing at each setting. We estimate standard errors for the conditional SD using the influence function of the conditional variance and the delta method for its square root~\cite[Ch.~3]{vandervaart1998asymptotic}, retaining the pairing for differences between thresholds.

\subsection{Joint and Marginal CDFs}
\label{subsec:numerical_cdfs}
\begin{figure*}[t]
 \centering
 \includegraphics[width=\textwidth]{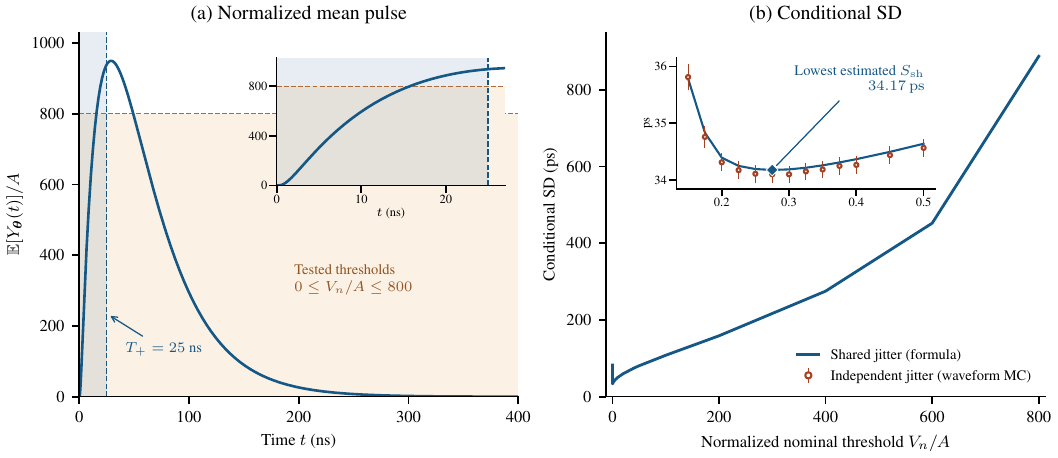}
 \caption{(a) Normalized mean pulse \(\Exp[Y_{\bm\theta}(t)]/A\), computed by convolving the arrival intensity \(\lambda_{\bm\theta}\) with the normalized response \(h/A\) in \eqref{eq:scintillation_example_model}. Blue and orange shading mark the observation window and tested threshold range, respectively. The inset enlarges the leading edge. (b) Conditional SD versus normalized nominal threshold \(V_n/A\). The line gives the shared-jitter formula estimates at all 33 thresholds. In the inset, open circles show independent-jitter waveform MC estimates with error bars of \(\pm1.96\) estimated standard errors. The diamond marks the smallest estimated \(S_{\mathrm{sh}}\) among the tested thresholds.}
 \label{fig:bg_threshold_sd}
\end{figure*}

The joint CDFs in Fig.~\ref{fig:bg_cdfs}(a) are evaluated on a \(23\times23\) grid with \(20\,\mathrm{ps}\) spacing in each cutoff. The mean pulse first reaches \(40A\) and \(50A\) at times that round to \(1.40\) and \(1.56\,\mathrm{ns}\), respectively, on the quantization grid. At these cutoffs, the formula and waveform MC estimates are \(0.4837\) and \(0.4845\), respectively, each with estimated standard error \(0.0016\). The joint formula estimate exceeds the product of the two marginal formula estimates (\(0.2982\)) by \(0.1855\).

The maximum absolute difference between the formula and waveform MC estimates is \(0.00431\) on the joint grid and \(0.00246\) over all \(1279\) marginal grid points at \(50A\) [Fig.~\ref{fig:bg_cdfs}(b)]. Hoeffding's inequality~\cite{hoeffding1963probability} for the difference of independent bounded sample means and a union bound over each grid give simultaneous \(95\%\) Monte Carlo sampling error bounds of \(0.00998\) and \(0.01042\), respectively. In \(768\) paired acquisitions, halving the waveform mesh spacing to \(5\,\mathrm{ps}\) changed the joint CDF formula estimates by at most \(8.5\times10^{-5}\) and the marginal CDF formula estimates at \(50A\) by at most \(2.0\times10^{-5}\).

To examine the error bound in \eqref{eq:independent_shared_crossing_bound} for replacing independent jitter by shared jitter, we estimated \(F(t)-F_{\mathrm{sh}}(t)\) and \(B_n(t)\) on the same analog crossings at \(0.05A\). The estimated CDF difference reached its largest magnitude at \(t=0.22\,\mathrm{ns}\), where it was \(4.36\times10^{-4}\) and the estimated bound was \(4.60\times10^{-4}\).

\subsection{Conditional SD versus Nominal Threshold}
\label{subsec:numerical_sd}
At \(50A\), \eqref{eq:bg_shared_sd} gave an estimate of \(80.77\,\mathrm{ps}\) for \(S_{\mathrm{sh}}\) on the curve in Fig.~\ref{fig:bg_threshold_sd}(b) (estimated standard error \(0.17\,\mathrm{ps}\)). The waveform MC estimate of \(S\) was \(80.80\,\mathrm{ps}\) (standard error \(0.18\,\mathrm{ps}\)).

In the main samples, both methods gave their smallest estimated conditional SD at \(0.275A\), with \(34.17\,\mathrm{ps}\) from the formula and \(34.10\,\mathrm{ps}\) from waveform MC. In a further \(8192\) independent acquisitions, waveform MC gave \(34.60\,\mathrm{ps}\) at \(0.275A\) and \(34.65\,\mathrm{ps}\) at \(0.25A\). The paired difference computed from the unrounded estimates was \(0.04\,\mathrm{ps}\) (higher at \(0.25A\)), with standard error \(0.09\,\mathrm{ps}\), leaving these two thresholds unresolved.

At \(800A\), the waveform MC estimate of \(S\) from the main sample was \(889.29\,\mathrm{ps}\). Extending the observation and recording windows to \([0,40]\,\mathrm{ns}\) and \((-0.3,40.3)\,\mathrm{ns}\) preserved the first recorded crossing time at both \(0.275A\) and \(800A\) in all \(1024\) paired acquisitions.

To limit false triggers, consider a minimum threshold of three times the measured baseline noise SD. For a cerium-doped LYSO crystal coupled to a SiPM at \(1.54\,\mathrm{V}\) overvoltage, Seifert et al.~\cite[Table~II]{seifert2012comprehensive} measured a baseline noise SD of \(0.32\,\mathrm{mV}\) and a mean single-cell amplitude of \(1.77\,\mathrm{mV}\). Matching the peak \(0.811A\) of \(h\) to that amplitude gives \(A\approx2.18\,\mathrm{mV}\). The minimum threshold is then \(0.96\,\mathrm{mV}\approx0.44A\) above baseline, excluding \(0.275A\). Among the eligible settings, the predicted conditional SD is smallest at \(0.45A\) and increases at each tested threshold through \(800A\) [Fig.~\ref{fig:bg_threshold_sd}(b)].

\section{Conclusion}
\label{sec:conclusion}

Averaging the composition of the crossing, timing, and selection kernels over the signal law gives the law of recorded crossings. When timing jitter is shared within each threshold, path infima and suprema determine the joint CDF of first recorded crossing times under the path and boundary conditions of Theorem~\ref{thm:joint_recorded_cdf}(i). At thresholds with positive recording probability, the marginal CDF determines the conditional SD. The error bounds cover both CDFs and the conditional SD under independent per-crossing jitter.

In the biexponential--Gaussian scintillation example, the shared-jitter predictions for the joint and marginal CDFs and the conditional SD agreed with independent-jitter waveform simulations within Monte Carlo sampling uncertainty. At the grid times nearest the mean pulse crossings at \(40A\) and \(50A\), the estimated joint probability was \(0.484\), compared with \(0.298\) from the product of the marginals. Across 33 thresholds, both methods gave their smallest estimated conditional SD at \(0.275A\) (about \(34\,\mathrm{ps}\)), although a separate simulation could not resolve the difference from \(0.25A\). When the minimum threshold was set to three times the measured baseline noise SD to limit false triggers, \(0.45A\) gave the smallest predicted conditional SD among the eligible settings, only \(0.32\,\mathrm{ps}\) above the minimum for the full scan.

\section*{Acknowledgment}

The authors thank Junhao Yu and Guanlin Su for helpful discussions. The authors used OpenAI GPT to help develop the probabilistic framework and derive the results in Sections II--IV and the Supplemental Material. They also used it to design simulations, write code, prepare figures, and interpret results in Section V, and to draft and edit the text throughout.

\makeatletter
\begingroup
\edef\@currentlabel{\arabic{equation}}\label{mt-main-equation-count}
\edef\@currentlabel{\arabic{definition}}\label{mt-main-definition-count}
\edef\@currentlabel{\arabic{proposition}}\label{mt-main-proposition-count}
\edef\@currentlabel{\arabic{theorem}}\label{mt-main-theorem-count}
\edef\@currentlabel{\arabic{corollary}}\label{mt-main-corollary-count}
\edef\@currentlabel{\arabic{lemma}}\label{mt-main-lemma-count}
\edef\@currentlabel{\arabic{figure}}\label{mt-main-figure-count}
\edef\@currentlabel{\arabic{table}}\label{mt-main-table-count}
\endgroup
\makeatother

\nocite{loignonhoule2023timing}
\bibliographystyle{IEEEtran}
\bibliography{MT_sampler-ref}

\end{document}


\setlength{\abovedisplayskip}{5.25pt}
\setlength{\belowdisplayskip}{\abovedisplayskip}
\setlength{\abovedisplayshortskip}{\abovedisplayskip}
\setlength{\belowdisplayshortskip}{\abovedisplayskip}
\raggedbottom
\clubpenalty=10000
\widowpenalty=10000
\displaywidowpenalty=10000
\interdisplaylinepenalty=10000
\AddToHook{env/definition/before}{\Needspace{4\baselineskip}}
\AddToHook{env/proof/before}{\Needspace{3\baselineskip}}
\newenvironment{titledproof}[1]{\begin{proof}[\nopunct]\textit{#1.}\enspace\ignorespaces}{\end{proof}}
\newcommand{\SIpropositionlines}{11}
\AddToHook{env/proposition/before}{\Needspace{\SIpropositionlines\baselineskip}}
\setlength{\emergencystretch}{5em}

\title{Supplemental Material for\\ Multi-Threshold Sampling: Signal Space, Sampling Operators, and First Recorded Crossing Time}
\author{Ao Qiu and Qingguo Xie}
\maketitle

\setcounter{equation}{\getrefnumber{mt-main-equation-count}}
\setcounter{definition}{\getrefnumber{mt-main-definition-count}}
\setcounter{proposition}{\getrefnumber{mt-main-proposition-count}}
\setcounter{theorem}{\getrefnumber{mt-main-theorem-count}}
\setcounter{corollary}{\getrefnumber{mt-main-corollary-count}}
\setcounter{lemma}{\getrefnumber{mt-main-lemma-count}}
\setcounter{figure}{\getrefnumber{mt-main-figure-count}}
\setcounter{table}{\getrefnumber{mt-main-table-count}}

This supplement uses the notation of the main article. Equation and result numbers continue those of the main article; citations refer to its reference list.

\section{Exact Marginal and Joint Distributions under Independent Timing Jitter}
\label{sec:supp_independent_cdf}

Fix a threshold \(n\) and finite cutoff \(t\). For a random signal as in Definition~\ref{def:signal_path_space} sampled according to Definition~\ref{def:shared_noise_mt_sampler}, write \(s_1,\ldots,s_k\) for the atoms of \(C_n\), with \(k=C_n(\R)\). Given \(C_n\), each crossing at \(s_i\) misses \(W\cap(-\infty,t]\) after timing with probability \(1-w_t(s_i)\). As in the proof of Theorem~\ref{thm:joint_recorded_cdf}, independence across crossings gives the marginal cumulative distribution function (CDF)
\[
 F(t)=1-\Exp\!\left[\prod_{i=1}^{k}(1-w_t(s_i))\right],
\]
with empty product equal to one.

To average over threshold levels, set \(X_n=Y-D_n\), where \(D_n(s)=\eta_n(s)-\eta_n(T_-)\) as in Sec.~\ref{subsubsec:joint_distribution_across_thresholds} of the main article. The crossings of \(X_n\) at \(\widetilde V_n(T_-)=V_n+\eta_n(T_-)\) have the same times and marks as those of \(Y\) at \(\widetilde V_n\). Given \(Y,D_n\), \(\eta_n(T_-)\) has law \(\kappa(D_n,\cdot)\).

Suppose that, almost surely, \(X_n\) is \(C^1\) on each open subinterval of a finite partition of \(\mathcal T\) chosen measurably from \((Y,D_n)\). On each such piece, \(X_n\) and its derivative extend continuously to the endpoints from within. Each connected component \(B\) of \(\{X_n'\ne0\}\) within a piece is an interval of strict monotonicity.

Let \(\bm J^*=(J_i^*)_{i\ge1}\) be a sequence of independent draws from \(P_J\), independent of \((Y,\bm\eta)\). Using the indices constructed measurably from \((Y,D_n)\) in the proof, assign distinct coordinates \(J_B^*\) and \(J_{\mathrm{pt}}^*(v)\) to the components \(B\) and interior partition points \(v\).

With \(A_t\) from \eqref{eq:timed_crossing_window}, let \(\mathcal V_{n,t}\) be the countable union of the images \(X_n(B\cap A_t(J_B^*))\) and the following intervals. Each interior partition point \(v\in A_t(J_{\mathrm{pt}}^*(v))\) contributes the interval between the one-sided limits of \(X_n\) at \(v\). Include an endpoint \(z\) exactly when \(X_n-z\) has opposite strict signs immediately to either side of \(v\), even when the limits coincide.

\begin{theorem}[Marginal CDF from a union of amplitude intervals]
\label{thm:supp_independent_interval_cdf}
Under the piecewise \(C^1\) assumptions above, suppose that, almost surely, \(X_n'(s)\ne0\) at every analog crossing at threshold \(n\) with direction mark \(+1\) or \(-1\) in the interior of a partition subinterval. Then for every \(t\in\R\),
\begin{equation}
 F(t)=\Exp\!\left[\kappa(D_n,\mathcal V_{n,t}-V_n)\right],
 \label{eq:supp_independent_interval_cdf}
\end{equation}
where \(\mathcal V_{n,t}-V_n=\{z-V_n:z\in\mathcal V_{n,t}\}\).
\end{theorem}

\begin{proof}
Given \(Y,\bm\eta\), each component \(B\) contains at most one crossing of the level \(\widetilde V_n(T_-)\). The image \(X_n(B\cap A_t(J_B^*))\) contains this level exactly when such a crossing lies in \(A_t(J_B^*)\). At a partition point \(v\), \eqref{eq:crossing_mark} gives a mark of \(\pm1\) exactly when \(X_n-\widetilde V_n(T_-)\) has opposite strict signs immediately to the left and right of \(v\). By the nonzero derivative assumption, the components and partition points account for every crossing with mark \(\pm1\). The shifts assigned to these crossings are independent with law \(P_J\), as in Definition~\ref{def:shared_noise_mt_sampler}. Thus \(\{\widetilde V_n(T_-)\in\mathcal V_{n,t}\}\) has the same conditional probability as recording a crossing by \(t\).

Encode the measurable partition by its size and ordered endpoints. Within a piece, \(X_n\) has jointly measurable derivative obtained by taking forward differences with step \(\min\{1/k,(v-s)/2\}\) at \(s\), where \(v\) is the right endpoint. Use zero wherever the limit fails.

Fix an enumeration \((r_k)_{k\ge1}\) of \(\mathbb Q\cap\mathcal T\) and set \(k(B)=\min\{k:r_k\in B\}\). Each open component contains a rational time, and disjoint components have distinct indices.
Two rational times in one piece belong to the same component of \(\{X_n'\ne0\}\) exactly when \(|X_n'|\) has positive infimum on the closed interval between them. Continuity permits this test on rational times alone. It thus determines \(k(B)\) and the endpoints \(\inf(B\cap\mathbb Q),\sup(B\cap\mathbb Q)\) measurably.
Assign \(J_B^*=J_{2k(B)}^*\) to component \(B\) and \(J_{\mathrm{pt}}^*(v)=J_{2j-1}^*\) to the \(j\)-th interior partition point \(v\) in increasing order.

Monotonicity gives the image of \(B\cap A_t(J_B^*)\) from its one-sided endpoint values and the inclusion of the corresponding time endpoints. At a partition point, strict positivity immediately to the left means that some left neighborhood of rational radius within the piece has \(\inf(X_n-z)>0\) on every rational-endpoint compact subinterval in its interior. Negative and right-hand signs admit the same test. Membership in \(\mathcal V_{n,t}\) is thus jointly measurable in the inputs and \(z\).

Given \(Y,D_n,\bm J^*\), the set \(\mathcal V_{n,t}\) is fixed and \(\eta_n(T_-)\) has conditional law \(\kappa(D_n,\cdot)\). Since \(\widetilde V_n(T_-)=V_n+\eta_n(T_-)\), integrating the membership indicator gives \eqref{eq:supp_independent_interval_cdf}.
\end{proof}

For nonempty \(\mathcal I\subseteq[N]\), construct \(\mathcal V_{n,t_n}\) as in Theorem~\ref{thm:supp_independent_interval_cdf}, using a separate sequence \(\bm J_n^*\) of independent draws from \(P_J\) at each \(n\in\mathcal I\). The sequences are independent across thresholds and jointly independent of \((Y,\bm\eta)\).

\begin{corollary}[Joint CDF from unions of amplitude intervals]
\label{cor:supp_independent_joint_cdf}
If the hypotheses of Theorem~\ref{thm:supp_independent_interval_cdf} hold at every \(n\in\mathcal I\), then for each \(\bm t=(t_n)_{n\in\mathcal I}\in\R^{\mathcal I}\),
\begin{equation}
 F_{\mathcal I}(\bm t)
 =\Exp\!\left[\prod_{n\in\mathcal I}
 \kappa\bigl(D_n,\mathcal V_{n,t_n}-V_n\bigr)\right].
 \label{eq:supp_independent_joint_cdf}
\end{equation}
\end{corollary}

\begin{proof}
Given \(Y,\bm\eta\), recording events are independent across thresholds because the jitter rows are independent and selection acts separately at each threshold. The events \(\{\eta_n(T_-)\in\mathcal V_{n,t_n}-V_n\}\) are likewise conditionally independent by independence of the auxiliary sequences. The proof of Theorem~\ref{thm:supp_independent_interval_cdf} identifies their conditional probabilities with those of the corresponding recording events. Averaging the product gives \(F_{\mathcal I}(\bm t)\).

Now condition on \(Y\), the increments \(D_n\) for \(n\in\mathcal I\), and the auxiliary sequences. The sets \(\mathcal V_{n,t_n}\) are fixed. Independence of the threshold noise paths from one another, \(Y\), and the auxiliary sequences gives the initial values the conditional joint law \(\bigotimes_{n\in\mathcal I}\kappa(D_n,\cdot)\). Integrating the membership indicators gives \eqref{eq:supp_independent_joint_cdf}.
\end{proof}

\section{Signal Paths, Sampler Measurability, and Admissibility}
\label{sec:supp_signal_sampler_measurability}

\begin{titledproof}{Verification of the sampling operators in Definition~\ref{def:shared_noise_mt_sampler}}
\emph{Crossing.} The map in \eqref{eq:shared_noise_crossing_operator} is jointly measurable by the construction of the stochastic crossing operator in \eqref{eq:crossing_realization_map}.

\emph{Timing.} Equip \(\R^{[N]\times\{1,2,\ldots\}}\) with the product Borel \(\sigma\)-algebra. Fix the counts \(\nu(\R\times\{n\}\times\{\sigma\})\) for all pairs \((n,\sigma)\). For each such pair and each integer \(1\le r\le\nu(\R\times\{n\}\times\{\sigma\})\), the time coordinate of the \(r\)-th atom is at most \(u\) exactly when \(\nu((-\infty,u]\times\{n\}\times\{\sigma\})\ge r\). This makes the ordered times measurable for each pair \((n,\sigma)\). Sorting these atoms at each threshold by time and then by direction mark gives measurable \(s_{n,k}\) and \(\sigma_{n,k}\).

For every Borel \(B\subseteq\mathsf E\),
\begin{equation*}
  \mathcal Q_{\bm j}(\nu)(B)
  =\sum_{n=1}^{N}\sum_{k=1}^{k_n(\nu)}
    \ind_B\bigl(r_\Delta(s_{n,k}+j_{n,k}),n,\sigma_{n,k}\bigr).
\end{equation*}
This finite sum is jointly measurable in \((\nu,\bm j)\) for fixed counts. Each count is measurable, and together the counts take only countably many values. Thus \eqref{eq:point_measure_sigma_algebra} gives joint measurability of \((\nu,\bm j)\mapsto\mathcal Q_{\bm j}(\nu)\). Each atom retains its threshold index and direction mark, so \(\mathcal Q_{\bm j}\) preserves \((\pi_{\mathrm m})_{\#}\nu\). Since the coordinates of \(\bm J\) are measurable, \(\Phi_{\mathcal Q}\) satisfies Definition~\ref{def:stochastic_timing_operator}.

\emph{Selection.} The restriction \(\nu\mapsto\nu_W^{\pm}\) is measurable by evaluation. For each fixed total mass \(m\), the lexicographically ordered atoms of \(\nu_W^{\pm}\) are measurable functions of \(\nu_W^{\pm}\), and the retained indices are determined by finitely many comparisons. Hence \(\mathcal H_W\) is measurable and, since \(\mathcal H_W(\nu)\leq\nu_W^{\pm}\leq\nu\), is a selection operator.
\end{titledproof}

\begin{titledproof}{Measurability of \(\tau_{n,W}\)}
Equip \(W\cup\{\infty\}\) with the trace of the Borel \(\sigma\)-algebra of the extended real line. For the map \(\tau_{n,W}\) in Definition~\ref{def:selected_crossing_time} and each finite \(t\), the sublevel set is
\begin{equation*}
  \{\nu:\nu((W\cap(-\infty,t])\times\{n\}\times\Sigma)>0\}.
\end{equation*}
This set is measurable by \eqref{eq:point_measure_sigma_algebra}, so \(\tau_{n,W}\) is measurable.
\end{titledproof}

\begin{titledproof}{The photon-arrival signal belongs to \(\mathfrak Y\)}
Enumerate \([k,k+1)\), \(k\in\mathbb Z\), in a fixed order and list the arrivals in each by time with multiplicity, keeping each arrival \(U_i\) paired with its single-photon response \(H_i\). The ordered times are measurable by count evaluations, and path evaluation makes each shifted response measurable at fixed \(t\). By Definition~\ref{def:poisson_photon_signal}, \(\sum_i\sup_{|t|\le m}|H_i(t-U_i)|<\infty\) almost surely for each positive integer \(m\). On the intersection of these probability-one events the sum converges locally uniformly and has right-continuous paths with finite left limits. Taking \(L=\mathcal T\) gives \(\Exp\sup_{t\in\mathcal T}|Y_{\bm\theta}(t)|<\infty\), as required.
\end{titledproof}

\begin{titledproof}{Integrability in the biexponential--Gaussian model}
In the biexponential--Gaussian model of Definition~\ref{def:biexponential_gaussian_model}, replace \(u\) by \(u/\ell_h\) in the Gaussian series from the proof of Corollary~\ref{cor:bg_joint}. On \([0,1]\), the expected supremum of its absolute value is at most \(\sum_{k\ge0}\Exp|\xi_k|/(\ell_h^k\sqrt{k!})<\infty\). By stationarity the same bound holds on every unit interval. The law of \(H_i-h\) is obtained by multiplying this process by \(s_h h\). Since \(h\) vanishes for \(u<0\) and \(0\le h(u)\le Ae^{-u/\tau_{\mathrm d}}\) for \(u\ge0\), summing the bounds over \([k,k+1]\), \(k\ge0\), gives \(\Exp\sup_{u\in\R}|H_i(u)|<\infty\). For every compact interval \(L\), independence of the responses and arrivals and \(\Exp N_{\mathrm{ph}}=\Lambda\) give
\[
 \Exp\!\left[\sum_{i=1}^{N_{\mathrm{ph}}}\sup_{t\in L}|H_i(t-U_i)|\right]
 \le \Lambda\,\Exp\sup_{u\in\R}|H_1(u)|<\infty .
\]
This gives the integrability required by Definition~\ref{def:poisson_photon_signal}. Applying the same unit-interval bound with \(\ell_\eta\) and scale \(s_\eta\) on the compact window \(\mathcal T\) gives \(\Exp_{\mathcal C}\max_{n\in[N]}\lVert\eta_n\rVert_\infty<\infty\) since \(N\) is finite.
\end{titledproof}

\begin{titledproof}{Almost-sure admissibility in the biexponential--Gaussian model}
For the admissible domain in Definition~\ref{def:shared_noise_mt_sampler}, take the continuous path pairs \((x,\mathbf v)\) for which each \(x-v_n\) has finitely many zeros on \(\mathcal T\). These zeros give a partition as in \eqref{eq:finite_state_partition}. Continuous paths form a Borel subspace of \(\mathsf X\) on which the Skorokhod \(J_1\) and uniform topologies agree. For each \(n\) and integer \(m\ge1\), \(x-v_n\) has at least \(m\) distinct zeros exactly when it has \(m\) zeros pairwise at least \(1/k\) apart for some integer \(k\ge1\). For fixed \(m,k\), the separation condition is closed by compactness of \(\mathcal T\). Taking the union over \(k\) gives a Borel set for each \(m\), so the chosen domain is Borel.

For the random signal and threshold paths of Definition~\ref{def:biexponential_gaussian_model}, \(Y(\omega)(s)=\widetilde V_n(\omega_{\mathcal C})(s)\) exactly when \(L_n(\omega,\omega_{\mathcal C})(s)=Z_n(\omega_{\mathcal C})\). As in the proof of Corollary~\ref{cor:bg_joint}, conditional on \(\bm U,Y\) and the residuals \(\eta_n^\perp\), the standard normal \(Z_n\) almost surely avoids the values of the constant pieces of \(L_n\). Each nonconstant piece extends analytically across its endpoints and hence meets \(Z_n\) only finitely often~\cite[Sec.~1.2]{krantz2002primer}. There are finitely many pieces and thresholds, so \(\bigl(Y(\omega),\widetilde{\mathbf V}(\omega_{\mathcal C})\bigr)\in\mathsf D_{\mathcal C}\) for \((\Prob\otimes\Prob_{\mathcal C})\)-almost every \((\omega,\omega_{\mathcal C})\in\Omega\times\Omega_{\mathcal C}\).
\end{titledproof}

\Needspace{8\baselineskip}
\section{Onsets, Time Differences, and Coincidence Time Resolution}
\label{subsec:arrival_difference_ctr}

\begin{definition}[Waveform onset]
\label{def:waveform_onset}
For a random signal \(Y\in\mathfrak Y\) as in Definition~\ref{def:signal_path_space}, the waveform onset relative to the zero baseline on \(\mathcal T\) is
\[
 t_{\mathrm{sig}}(Y)\triangleq
 \inf\{t\in\mathcal T:\Prob\{Y(t)\ne 0\}>0\},
\]
with \(\inf\varnothing=+\infty\).
\end{definition}

The waveform onset is determined by the law of \(Y\) and is common to all representatives in Definition~\ref{def:signal_path_space}.

\begin{definition}[Photon-arrival onset]
\label{def:photon_arrival_onset}
When the following infimum is finite, the photon-arrival onset of the process in Definition~\ref{def:poisson_photon_signal} is
\begin{equation}
 t_0\triangleq\inf\left\{t\in\R:
        \Prob\{\Pi_{\bm\theta}((-\infty,t])>0\}>0\right\}.
 \label{eq:photon_arrival_onset}
\end{equation}
\end{definition}

\begin{proposition}[Equality of waveform and photon-arrival onsets]
\label{prop:waveform_arrival_onsets}
Let \(Y=Y_{\bm\theta}|_{\mathcal T}\) be given by Definition~\ref{def:poisson_photon_signal}, with photon-arrival onset \(T_-\le t_0<T_+\). Let \(H\) have the common law of the single-photon responses in that definition. Assume that \(H\) almost surely vanishes on \((-\infty,0)\) and that its finite mean \(h(u)=\Exp[H(u)]\) satisfies \(h(0)\ge0\) and \(h(u)>0\) for \(0<u<\delta\), for some \(\delta>0\). Then \(t_{\mathrm{sig}}(Y)=t_0\).
\end{proposition}

\begin{proof}
For each rational \(q<t_0\), the event of an arrival by \(q\) has probability zero by \eqref{eq:photon_arrival_onset}. Taking the union over these cutoffs shows that almost surely no arrival precedes \(t_0\). Causality then gives \(Y(t)=0\) almost surely for \(t<t_0\) in \(\mathcal T\).

Fix \(t_0<t<\min\{T_+,t_0+\delta\}\). Conditional on the arrivals, the mean of \(Y(t)\) is the sum of \(h(t-U_i)\) over \(U_i\in[t_0,t]\), by the common response law and causality. This sum is finite by local finiteness and is integrable by Definition~\ref{def:poisson_photon_signal}. Its terms are nonnegative, and a term is positive whenever \(U_i<t\). By \eqref{eq:photon_arrival_onset}, an arrival occurs before \(t\) with positive probability. Hence \(\Exp[Y(t)]>0\), so \(\Prob\{Y(t)\ne0\}>0\). Such times approach \(t_0\) from above, proving the equality.
\end{proof}

\begin{corollary}[Onsets in the biexponential--Gaussian model]
\label{cor:bg_waveform_onset}
Under Definition~\ref{def:biexponential_gaussian_model}, if \(T_-\le t_0<T_+\), both onsets equal the parameter \(t_0\).
\end{corollary}

\begin{proof}
For Poisson arrivals, a positive count by time \(t\) has positive probability exactly when \(\int_{-\infty}^t\lambda_{\bm\theta}(u)\,du>0\). By \eqref{eq:scintillation_example_model}, this holds exactly for \(t>t_0\). The same equation gives \(h(0)=0\) and \(h(u)>0\) for every \(u>0\). At each negative rational delay the response has zero mean and zero variance. Continuity gives \(H(u)=0\) for all \(u<0\) almost surely. Proposition~\ref{prop:waveform_arrival_onsets} now applies.
\end{proof}

Consider two independent detector channels, each with a random signal \(Y^{(i)}\) as in Definition~\ref{def:signal_path_space} and a sampler as in Definition~\ref{def:shared_noise_mt_sampler}. Use a common observation window \(\mathcal T\), recording window \(W=(a,b)\), and threshold index \(n\). Write \(t_{\mathrm{sig}}^{(i)}=t_{\mathrm{sig}}(Y^{(i)})\).

In channel \(i\), the first recorded crossing time \(\tau_n^{(i)}\) has CDF \(F^{(i)}\), recording probability \(p^{(i)}=F^{(i)}(b)>0\), and standard deviation (SD) \(S^{(i)}\) conditional on recording. The time difference is defined on \(R=\{\tau_n^{(1)}<\infty,\ \tau_n^{(2)}<\infty\}\), where both channels record a crossing, under the law \(\Prob(\cdot\mid R)\).

\begin{proposition}[Conditional unbiasedness of the time difference]
\label{prop:unbiased_time_difference}
If both waveform onsets are finite, then conditional on \(R\), \(\tau_n^{(1)}-\tau_n^{(2)}\) is unbiased for \(t_{\mathrm{sig}}^{(1)}-t_{\mathrm{sig}}^{(2)}\) if and only if
\begin{equation}
 t_{\mathrm{sig}}^{(1)}+\frac{\int_a^b F^{(1)}(t)\,dt}{p^{(1)}}
 =t_{\mathrm{sig}}^{(2)}+\frac{\int_a^b F^{(2)}(t)\,dt}{p^{(2)}}.
 \label{eq:time_difference_unbiasedness}
\end{equation}
\end{proposition}

\begin{proof}
Independence gives, for finite \(s,t\),
\begin{equation*}
 \Prob\bigl(\tau_n^{(1)}\leq s,\tau_n^{(2)}\leq t\mid R\bigr)
 =\frac{F^{(1)}(s)F^{(2)}(t)}{p^{(1)}p^{(2)}}.
\end{equation*}
For each channel \(i\), the conditional CDF given \(R\) is \(F^{(i)}(t)/p^{(i)}\) at finite \(t\). Since a recorded time lies in \((a,b)\), Tonelli's theorem gives
\begin{align*}
 \Exp[\tau_n^{(i)}\mid R]
 &=\Exp[\tau_n^{(i)}\mid\tau_n^{(i)}<\infty]\\
 &=b-\frac{1}{p^{(i)}}\int_a^b F^{(i)}(t)\,dt.
\end{align*}
Subtracting the two means gives
\begingroup
\begin{gather}
 \Exp[\tau_n^{(1)}-\tau_n^{(2)}\mid R]
 =\int\limits_a^b\!\left(\frac{F^{(2)}(t)}{p^{(2)}}-\frac{F^{(1)}(t)}{p^{(1)}}\right)dt.
 \label{eq:time_difference_mean}
\end{gather}
\endgroup
Equating this with \(t_{\mathrm{sig}}^{(1)}-t_{\mathrm{sig}}^{(2)}\) proves the claim.
\end{proof}

If the signals have the same law and the sampler settings and noise laws agree, the waveform onsets agree and \(F^{(1)}=F^{(2)}\), so \eqref{eq:time_difference_unbiasedness} holds whenever the common onset is finite.

\begingroup
\renewcommand{\SIpropositionlines}{4}
\begin{proposition}[Conditional variance of the time difference]
\label{prop:time_difference_variance}
For the two independent detector channels,
\begin{equation}
 \operatorname{Var}\bigl(\tau_n^{(1)}-\tau_n^{(2)}\mid R\bigr)=\bigl(S^{(1)}\bigr)^2+\bigl(S^{(2)}\bigr)^2.
 \label{eq:time_difference_variance}
\end{equation}
\end{proposition}
\endgroup

\begin{proof}
The conditional joint law factors as in the proof of Proposition~\ref{prop:unbiased_time_difference}. The two times are bounded on this event and their conditional covariance is zero, giving \eqref{eq:time_difference_variance}.
\end{proof}

\begin{definition}[SD-based coincidence time resolution]
\label{def:ctr}
The SD-based coincidence time resolution (CTR) on the Gaussian full width at half maximum scale~\cite[Sec.~II-B]{loignonhoule2023timing} is
\begin{equation}
 \mathrm{CTR}\triangleq2\sqrt{2\ln2}\,
 \sqrt{\operatorname{Var}\bigl(\tau_n^{(1)}-\tau_n^{(2)}\mid R\bigr)}.
 \label{eq:ctr_definition}
\end{equation}
\end{definition}

Substituting \eqref{eq:time_difference_variance} into \eqref{eq:ctr_definition} gives
\begin{align}
 \mathrm{CTR}&=2\sqrt{2\ln2}\sqrt{\bigl(S^{(1)}\bigr)^2+\bigl(S^{(2)}\bigr)^2}
 \label{eq:ctr_sd_relation}\\
 &=4\sqrt{\ln2}\,S \qquad (S^{(1)}=S^{(2)}=S).
 \label{eq:ctr_identical_sd}
\end{align}
If the two SDs agree at every setting, CTR and \(S\) have the same minimizers and percentage reductions from any reference setting with \(S>0\).